\documentclass[11pt]{article}
\usepackage{amsmath}
\usepackage{amssymb,amsbsy,amsthm}
\theoremstyle{plain}
\usepackage{graphicx}
\usepackage[dvipsnames]{xcolor}
\usepackage{enumerate}
\usepackage{cases}
\usepackage[margin=1in]{geometry}
\usepackage{pdfsync}
\usepackage{hyperref}

\allowdisplaybreaks[1]

\numberwithin{equation}{section}

\DeclareMathOperator{\R}{\mathbb{R}} 

\newcommand{\p}{\mathbb{P}} 
\newcommand{\E}{\mathbb{E}} 

\DeclareMathOperator{\Bin}{Bin} 

\DeclareMathOperator*{\argmax}{arg\,max} 
\DeclareMathOperator*{\argmin}{arg\,min} 

\newcommand{\wh}{\widehat}  
\newcommand{\wt}{\widetilde} 

\newcommand{\ML}{\mathrm{ML}} 

\newcommand{\vs}{\mathit{vs}} 

\def\cI{{\mathcal I}}

\newtheorem{theorem}{Theorem}
\newtheorem*{theorem*}{Theorem}

\begin{document}

\title{Rumor source detection with multiple observations \\ under adaptive diffusions}
\author{
	Mikl\'os Z.\ R\'acz
	\thanks{Princeton University; \texttt{mracz@princeton.edu}. Research supported in part by NSF grant DMS 1811724 and by a Princeton SEAS Innovation Award.} 
	\and
	Jacob Richey
	\thanks{University of Washington; \texttt{jfrichey@uw.edu}.}
}
\date{\today}

\maketitle


\begin{abstract} 
Recent work, motivated by anonymous messaging platforms, 
has introduced adaptive diffusion protocols 
which can obfuscate the source of a rumor: 
a ``snapshot adversary'' with access to the subgraph of ``infected'' nodes can do no better than randomly guessing the entity of the source node. 
What happens if the adversary has access to multiple independent snapshots? 
We study this question when the underlying graph is the infinite $d$-regular tree. 
We show that 
(1) a weak form of source obfuscation is still possible in the case of two independent snapshots, 
but (2) already with three observations there is a simple algorithm that finds the rumor source with constant probability, regardless of the adaptive diffusion protocol. 
We also characterize the tradeoff between local spreading and source obfuscation for adaptive diffusion protocols (under a single snapshot). 
These results raise questions about the robustness of anonymity guarantees when spreading information in social networks. 
\end{abstract}


\section{Introduction} \label{sec:intro} 

Detecting the source of information diffusion on a network is an important problem in network science, 
with applications such as 
finding the source of a virus epidemic or
finding the source of a rumor on Twitter.  
A prototypical graph on which source detection is studied is the infinite $d$-regular tree $\mathbb{T}_{d}$ (with $d\geq 3$), which is our focus in this paper as well.

\textbf{Rumor source detection.} 
Perhaps the simplest and most natural model of information diffusion on a network is the susceptible-infected (SI) model, where the rumor is spread along each edge of the network at a constant rate, and once a node is infected it remains infected forever. Shah and Zaman studied detecting the source in this model~\cite{shah2010detecting,sz11rumor}. Formally, at time $t=0$ a vertex $v^{*} \in \mathbb{T}_{d}$ is ``infected'' and the information propagates on the network according to the SI model; 
one then observes the subset $V_{t}$ of infected vertices at time $t$, which consists of $N_{t} := \left| V_{t} \right|$ vertices. We assume that the underlying graph (in this case $\mathbb{T}_{d}$) is known and hence the subgraph $G_{t}$ induced by the vertices in $V_{t}$ is also known. The goal is to find the rumor source~$v^{*}$.

The maximum likelihood estimator (MLE) 
$\widehat{v}_{\mathrm{ML}} := \argmax_{v \in V_{t}} \p \left( G_{t} \, \middle| \, v^{*} = v \right)$ 
has particularly nice properties in this setting~\cite{shah2010detecting,sz11rumor}. 
In particular, 
Shah and Zaman 
showed that it is computable in linear time and that it detects the source with constant probability. 
More precisely, they show (in~\cite{shah2016finding}) that there exists a universal constant $\alpha_{d} > 0$ such that 
$\lim_{t \to \infty} \p \left( \widehat{v}_{\mathrm{ML}} = v^{*} \right) =\alpha_{d}$ (when $d\geq 3$). 
Many results extend to more general settings such as random trees~\cite{shah2016finding}.

Wang et al.~\cite{wang2014} studied rumor source detection in the same setting but now with multiple independent observations; 
that is, observing the infected nodes 
$V_{t}^{(1)}, \ldots, V_{t}^{(k)}$ 
of $k$ independent diffusions started from the same source $v^{*}$. 
They show that the detection probability increases with $k$ and that it goes to $1$ exponentially as $k \to \infty$.

\textbf{Rumor source obfuscation.} 
The results above show that if information propagates according to the SI model, then the source can be found efficiently and with good probability (that is, with at least constant probability). 
In certain applications, such as anonymous messaging apps\footnote{Examples include Whisper~\cite{whisper}, Blind~\cite{blind}, and the now-defunct Yik Yak~\cite{yikyak} and Secret~\cite{secret}.}, this is undesirable. 
Motivated by these applications, 
Fanti~et~al.~\cite{fanti2015spy} asked whether it is possible to devise messaging protocols that can \emph{obfuscate} the rumor source, while at the same time still spreading information widely and quickly.

They devised a family of messaging protocols, 
termed \emph{adaptive diffusions}, 
for this purpose; see Section~\ref{sec:adaptive_diffusion} for a detailed description. 
Their main result shows that a specific messaging protocol within this family achieves \emph{perfect obfuscation}: 
under this spreading model a ``snapshot adversary'' 
can do no better than randomly guessing the source node: 
\begin{equation}\label{eq:perfect_obfuscation}
\p \left( \widehat{v}_{\mathrm{ML}} = v^{*} \, \middle| \, N_{t} = n \right) 
= \frac{1+o \left( 1 \right)}{n}.
\end{equation}
Many results extend to more general settings such as irregular trees~\cite{fanti2016irregular,fanti2017hide}. 

\textbf{Our results.} 
We study the source obfuscation guarantees that adaptive diffusion protocols can provide, in a couple of settings. 
First, we do this in the context of the adversary having multiple independent observations. 
We show that when an adversary has access to two independent observations then a weak form of obfuscation is still possible. 
However, when it has access to three or more independent snapshots, then source detection with constant probability is always possible, regardless of the adaptive diffusion protocol.

We also do this in the context of spreading information \emph{locally} around the source. 
We introduce a natural quantitative measure of local spreading, and characterize the tradeoff between local spreading and source obfuscation for adaptive diffusion protocols (under a single snapshot).

Put together, these results raise questions about the robustness of possible anonymity guarantees when spreading information in social networks. 
In order to precisely state our results, 
we first describe in Section~\ref{sec:adaptive_diffusion} the setting of information diffusion processes in general 
and adaptive diffusions in particular. 
We then state our results in Sections~\ref{sec:results} and~\ref{sec:local_spreading}.

\subsection{Information diffusion and adaptive diffusion} \label{sec:adaptive_diffusion} 

We define a (discrete time) information diffusion process on a graph $G = \left( V, E \right)$ as 
a (potentially random) increasing sequence of subgraphs 
$G_{0} \subseteq G_{1} \subseteq G_{2} \subseteq \ldots$, 
where $G_{t} = \left( V_{t}, E_{t} \right)$ is the subgraph induced by the vertices $V_{t}$ who have the information at time $t$. 
Throughout the paper we assume that $G_{0}$ consists of a single vertex $v^{*} \in G$, which we term the \emph{source}. 
We also assume that the information spreads along the edges of the graph and hence 
$V_{t+1} \subseteq V_{t} \cup \partial G_{t}$, 
where 
$\partial G_{t} := \left\{ v \in V : v \notin V_{t}, \exists w \in V_{t} : (v,w) \in E \right\}$ 
denotes the (outer) vertex boundary of $G_{t}$, 
consisting of vertices that are not in $G_{t}$ but which are connected to a vertex in $G_{t}$.

A simple example is when every vertex who obtains the information spreads it to all its neighbors in the next time step. 
In this case $G_{t} = B_{t} \left( v^{*} \right)$ for every $t \geq 0$, 
where 
$B_{r} \left( v \right) := \left\{ u \in V : \delta_{G} \left( u, v \right) \leq r \right\}$ 
denotes the (closed) ball of radius $r$ around vertex $v \in V$ 
(here $\delta_{G}$ denotes graph distance in $G$). 
The SI model mentioned above\footnote{Note that the SI model is often defined in continuous time. Viewing this continuous-time process at the times when a new vertex obtains the information, we obtain the described discrete time information diffusion process.} can be defined inductively as follows: 
given~$G_{t}$, let $v_{t+1}$ be a uniformly randomly chosen vertex from $\partial G_{t}$ 
and let $V_{t+1} := V_{t} \cup \left\{ v_{t+1} \right\}$.

The \emph{source detection} problem is the following: 
given the underlying graph $G$, 
the distribution of the sequence $\left\{ G_{t} \right\}_{t \geq 0}$, 
and a single observation $G_{t}$ at some time $t > 0$, 
the goal is to estimate the source $v^{*}$. 
This is also known as the ``snapshot adversary'' model, since we get to observe $G_{t}$, a single snapshot in time.

\emph{Adaptive diffusion}, introduced by Fanti~et~al.~\cite{fanti2015spy}, is a family of information diffusion processes designed with \emph{source obfuscation} in mind. 
We now introduce and define adaptive diffusion on 
$\mathbb{T}_{d} =: G$; 
we refer the reader to~\cite{fanti2017hide} for a comprehensive introduction 
more generally.

Adaptive diffusion is defined via an auxiliary process, 
the path $\left\{ \vs_{t} \right\}_{t \geq 0}$ of a so-called \emph{virtual source}.  
This is a time-inhomogeneous Markov chain, which we now define. 
Initially, the virtual source is the same as the true source: $\vs_{0} := v^{*}$. 
Next, 
it moves to a uniformly random neighbor of~$v^{*}$: 
\begin{equation*}\label{eq:adfirst_step}
\p \left( \vs_{1} = w \right) = \frac{1}{d} \mathbf{1}_{\left\{ (w,v^{*}) \in E \right\}}. 
\end{equation*}
For the remainder of the path, assuming that $\vs_{t}$ is given, $\vs_{t+1}$ is defined as follows. 
If $t$ is odd, then $\vs_{t+1} = \vs_{t}$; that is, the virtual source stays put. 
If $t$ is even, then the virtual source either stays put 
or it moves to one of its $d-1$ neighbors that it has not visited before; 
in the latter case, it chooses the neighbor to move to uniformly at random. 
Note that if the virtual source moves then it moves \emph{away} from the source $v^{*}$. 
The probability of choosing one action or the other is a function of time $t$ and also the distance of $\vs_{t}$ from $v^{*}$ (hence the name \emph{adaptive}). 
Specifically, let $h_{t} := \delta_{G}(\vs_{t}, v^{*})$ denote the graph distance between $\vs_{t}$ and $v^{*}$. 
Then 
\begin{itemize} 
\item with probability $\alpha(t,h_{t})$ we have that $\vs_{t+1} = \vs_{t}$, that is, the virtual source stays put;  
\item and with probability $1 - \alpha(t,h_{t})$ the virtual source moves to one of its $d-1$ neighbors that it has not visited before, chosen uniformly at random. 
\end{itemize}
The probabilities $\alpha(t,h) \in \left[ 0, 1 \right]$, with $t \in \left\{ 2, 4, 6, \ldots \right\}$ and $h \in \left\{ 1, 2, 3, \ldots, t/2 \right\}$, 
are parameters that fully describe the distribution of the path $\left\{ \vs_{t} \right\}_{t \geq 0}$ of the virtual source. 
Each choice of parameters defines a particular 
Markov chain 
and thus a particular adaptive diffusion protocol.

Having defined the path of the virtual source, we are now ready to define the associated adaptive diffusion protocol, given $\left\{ \vs_{t} \right\}_{t \geq 0}$. 
When $t$ is even, the set of infected nodes is defined as 
\begin{equation*}\label{eq:ad_def}
 V_{t} := \left\{ v \in V : \delta_{G} \left( v, \vs_{t} \right) \leq t/2 \right\}. 
\end{equation*}
\begin{figure}[h!]
\begin{minipage}[c]{0.47\textwidth}
\centering 
 \includegraphics[width=0.88\textwidth]{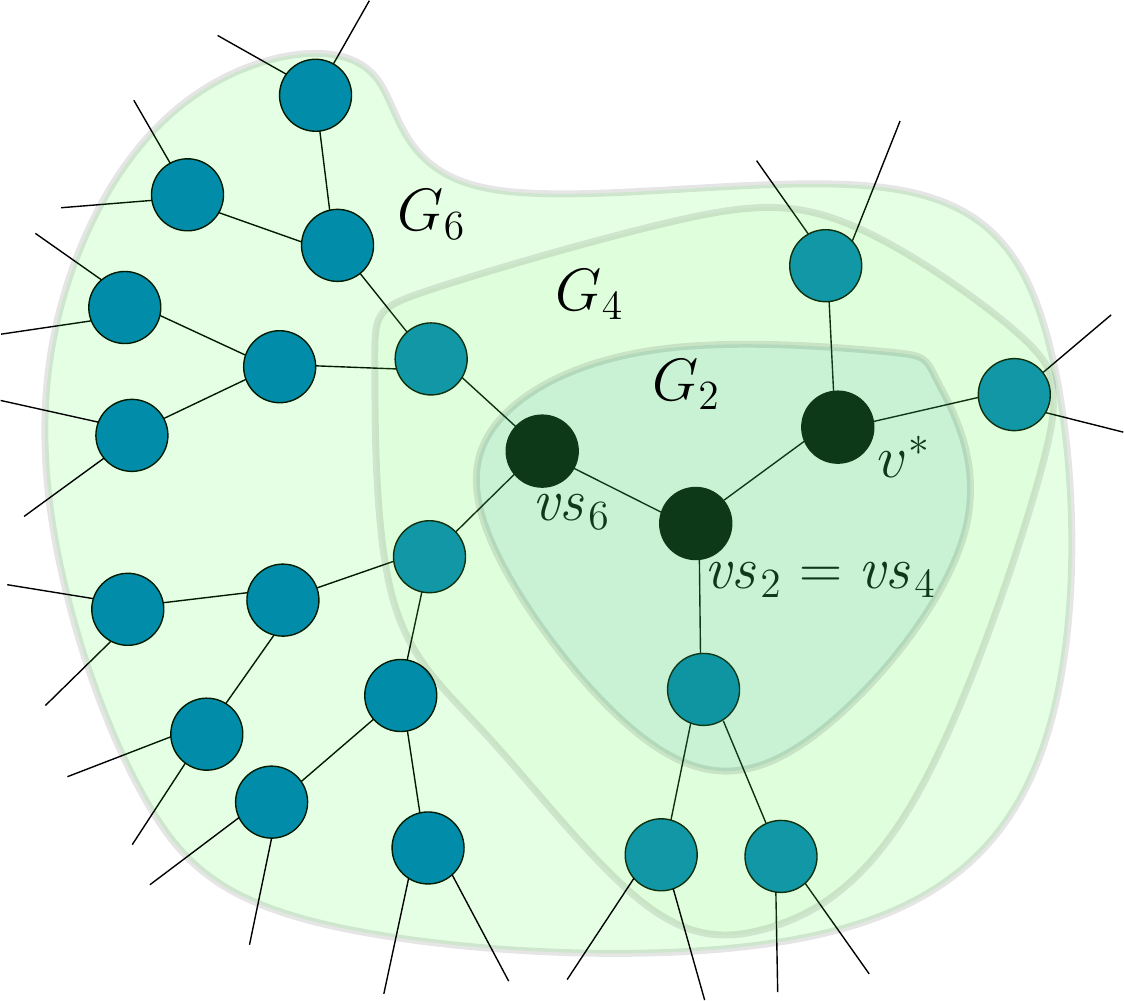} 
\end{minipage}\hfill
\begin{minipage}[c]{0.47\textwidth}
 \caption{An example of an adaptive diffusion spreading on the infinite $3$-regular tree~$\mathbb{T}_3$. 
 Here the virtual source moves to a uniformly randomly chosen neighbor of the source $v^{*}$ at time 1, then it stays put for several time steps, and moves again at time 5. The shaded regions show the infected subgraphs $G_{t}$ for $t \in \{2, 4, 6\}$; note that they are all balanced trees of depth~$t/2$, centered at the virtual source $\vs_{t}$.}
 \label{fig:ad}
\end{minipage}
\end{figure}
That is, $V_{t}$ is a ball of radius $t/2$---equivalently, a balanced tree of depth $t/2$---around the virtual source $\vs_{t}$. 
For $t$ odd, the set of infected nodes $V_{t}$ is chosen so that 
$\left\{ G_{t} \right\}_{t \geq 0}$ satisfies $V_{t+1} \subseteq V_{t} \cup \partial G_{t}$.\footnote{Specifically, we have the following. 
First, $V_{1} := \left\{ \vs_{0}, \vs_{1} \right\}$. 
Next, for $t \geq 3$ such that $t$ is odd, we distinguish two cases. 
If $\vs_{t} = \vs_{t-1}$, then $V_{t} := V_{t-1}$; 
that is, if the virtual source stays put (instead of moving), then the set of infected nodes is unchanged. 
If $\vs_{t} \neq \vs_{t-1}$, then 
$V_{t} := V_{t-1} \cup \left\{ w \in \partial G_{t-1} : \delta_{G} \left( w, \vs_{t} \right) = (t-1)/2 \right\}$; 
in other words, if the virtual source moves, then the information is spread in the same direction.} 
The resulting information diffusion process is called an \emph{adaptive diffusion}; 
see Figure~\ref{fig:ad} for an illustration. 
Note that by construction adaptive diffusion spreads the information to 
$N_{t} \asymp \left( d - 1 \right)^{t/2}$ 
nodes at time $t$, 
which is only a factor of two slower than the fastest possible spread.

Fanti~et~al.~\cite{fanti2015spy} show that a particular adaptive diffusion protocol---specifically, the process with   
$\alpha(t,h) := 
( \left( d - 1 \right)^{t/2 - h + 1} - 1 ) / 
( \left( d - 1 \right)^{t/2 + 1} - 1 )$---perfectly obfuscates the source from an adversary who sees a snapshot of a single diffusion. 
The key property of this construction is that, for $t$ even, 
all vertices in $V_{t} \setminus \left\{ \vs_{t} \right\}$ are equally likely to be the original source $v^{*}$ and hence an adversary can do no better than randomly guess among them. 
A similar statement holds also for $t$ odd, 
showing that the MLE satisfies~\eqref{eq:perfect_obfuscation}.

\subsection{Results: adaptive diffusion with multiple independent observations} \label{sec:results} 

In many applications it is common for individuals to send not just one but multiple messages over time, each one spreading over the same underlying network. If an adversary has access to a snapshot of each such diffusion, then they are in a much better position to find the source. 
Is it still possible to obfuscate the source with some form of information diffusion? We investigate this question in the context of adaptive diffusion protocols.

We show that when an adversary has access to two independent observations, a weak form of obfuscation is still possible with adaptive diffusion. However, when three or more independent observations are available, detection with constant probability is always possible, regardless of which adaptive diffusion protocol is used. This is the content of Theorems~\ref{thm:two_obs} and~\ref{thm:multi_obs}.

\begin{theorem}[Two independent observations]\label{thm:two_obs} 
Suppose that information is spread according to an adaptive diffusion protocol on $\mathbb{T}_d$, $d \geq 3$, and that an adversary has two independent observations of infected subgraphs, $G_{t_{1}}^{1}$ and $G_{t_{2}}^{2}$, started from a fixed source $v^{*}$. 
\begin{enumerate}[(a)] 
\item \label{thm:two_obs_detect} 
There exists a computationally efficient estimator $\widehat{v}$, which is agnostic to the adaptive diffusion protocol, such that if $t_{1}, t_{2} \geq 2$ then 
\begin{equation*}
\p \left(\widehat{v} = v^{*} \right) \geq \frac{d-1}{d} \cdot \frac{2}{\min \left\{ t_{1}, t_{2} \right\}}.
\end{equation*}
\item \label{thm:two_obs_obfu} 
There exists an adaptive diffusion protocol such that  the maximum likelihood estimator $\widehat{v}_{\ML}$ satisfies for all $t_{1}, t_{2} \geq 1$ that 
\begin{equation}\label{eq:two_obs_obfu}
\p\left(\widehat{v}_{\ML}  = v^{*} \right) \leq \frac{d-1}{d} \cdot \frac{7}{\min \left\{ t_{1}, t_{2} \right\}}.
\end{equation}
\end{enumerate}
\end{theorem}
A few comments are in order. 
First, the bounds in parts~\eqref{thm:two_obs_detect} and~\eqref{thm:two_obs_obfu} above match up to a small constant factor, hence this is best possible within the family of adaptive diffusions. 
Next, the detection probability in~\eqref{eq:two_obs_obfu} still vanishes as $t = \min \left\{ t_{1}, t_{2} \right\} \to \infty$, but only very slowly---exponentially more slowly than in the case of one observation (see~\eqref{eq:perfect_obfuscation} and recall that $N_{t} \asymp (d-1)^{t/2}$ is exponential in $t$). 
We also note that the adaptive diffusion protocol in part~\eqref{thm:two_obs_obfu} is different from the one used by Fanti~et~al.~\cite{fanti2015spy} to achieve perfect obfuscation in the case of a single observation; 
in fact, if this latter adaptive diffusion protocol is used to independently spread two diffusions, then the estimator $\widehat{v}$ in part~\eqref{thm:two_obs_detect} succeeds at finding the source with constant probability. 
Finally, we mention that the estimator in part~\eqref{thm:two_obs_detect} is essentially the same as the MLE in part~\eqref{thm:two_obs_obfu} when $t_{1}$ and $t_{2}$ are both even---see Section~\ref{sec:two_obs} for details.

Once the adversary has three independent observations, 
not even weak obfuscation is possible with adaptive diffusion. 
In fact, the detection probability converges to one exponentially quickly in the number of observations (see~\eqref{eq:exp_prob} below), 
extending the results of Wang~et~al.~\cite{wang2014} for the SI~model to the family of adaptive diffusions. 

\begin{theorem}[Three or more independent observations]\label{thm:multi_obs} 
Suppose that information is spread according to an adaptive diffusion protocol on $\mathbb{T}_d$, $d \geq 3$, and that an adversary has $k \geq 3$ independent observations of infected subgraphs, $G_{t_{i}}^{i}$ for $i \in \left\{ 1, \ldots, k \right\}$, started from a fixed source $v^{*}$. 

When $k = 3$, there is a computationally efficient estimator $\widehat{v}$, which is agnostic to the adaptive diffusion protocol, satisfying
\begin{equation*}
\p \left(\widehat{v} = v^{*}  \right) \geq \frac{(d-1)(d-2)}{d^2}.
\end{equation*}
More generally, 
there exists a computationally efficient estimator $\widehat{w} = \widehat{w}(k)$, which is agnostic to the adaptive diffusion protocol, such that 
\begin{equation}\label{eq:exp_prob}
\p \left(\widehat{w} = v^{*} \right) \geq 1 - d \times \exp \left( - \tfrac{\left( d - 2 \right)^{2}}{2d^{2}} k \right).
\end{equation}
\end{theorem} 

This result follows from basic symmetry properties of adaptive diffusion (see Section~\ref{sec:multi_obs}). 



\subsection{Results: local spreading vs.\ source obfuscation} \label{sec:local_spreading} 

It is often desirable to not only spread information widely and quickly, but also to spread it \emph{locally} around the source. Indeed, the local neighborhood of the source typically consists of nodes that are closely related to the source, and the information that the source is spreading is often most relevant to this local neighborhood. 
In particular, this is true for scenarios where source obfuscation is relevant and important, for instance, spreading information about a local protest. 
At the same time, local spreading is at odds with source obfuscation. 
Here we introduce a natural way to quantify local spreading, and characterize the tradeoff between local spreading and source obfuscation for adaptive diffusion protocols (under a single snapshot). 

Formally, define for an adaptive diffusion the quantity 
\[
R_{t} := \max \left\{ r \geq 0 : B_{r} \left( v^{*} \right) \subseteq G_{t} \right\}.
\]
In words, $R_{t}$ is the radius of the largest ball of infected nodes centered at the rumor source at time $t$. Since $R_{t}$ is (in general) a random quantity, we may use 
$\E \left[ R_{t} \right]$ 
as a deterministic measure of local spreading of an adaptive diffusion protocol. 
Observe that $0 \leq R_{t} \leq t/2$ and hence also $0 \leq \E \left[ R_{t} \right] \leq t/2$. 

Ideally for local spreading we would like 
$\E \left[ R_{t} \right]$ 
to grow linearly with $t$; 
at the very least, local spreading requires 
$\E \left[ R_{t} \right] \to \infty$ as $t \to \infty$. 
However, the adaptive diffusion protocol that achieves perfect source obfuscation (see the end of Section~\ref{sec:adaptive_diffusion}) does not have local spreading: 
in fact, 
$\E \left[ R_{t} \right] \leq 1$ for all $t$ and, moreover, $\sup_{t} R_{t}$ is finite almost surely. 

This shows that source obfuscation guarantees have to be relaxed in order to have local spreading. It turns out that it is still possible to have reasonable source obfuscation guarantees---we refer to this as ``polynomial obfuscation'', see~\eqref{eq:poly_obfuscation} below---and local spreading at the same time. 
The following theorem characterizes this tradeoff for adaptive diffusion protocols (under a single snapshot). For simplicity, we focus here on even times $t$.

\begin{theorem}[Tradeoff between local spreading and source obfuscation]\label{thm:local_spreading}
Suppose that information is spread according to an adaptive diffusion protocol on $\mathbb{T}_d$, $d \geq 3$, and that an adversary observes, at an even time $t$, an infected subgraph, $G_{t}$, started from a fixed source $v^{*}$. 
\begin{enumerate}[(a)]
\item Suppose that the adaptive diffusion protocol achieves ``polynomial obfuscation'', 
that is, the following holds: 
\begin{equation}\label{eq:poly_obfuscation}
\p \left( \wh{v}_{\ML} = v^{*} \right) 
\leq \frac{C}{N_{t}^{\gamma}} 
\end{equation}
for some $\gamma \in \left( 0, 1 \right)$ and $C < \infty$, 
where recall that 
\[
N_{t} = \left| V_{t} \right| = \tfrac{d}{d-2} \left( \left( d - 1 \right)^{t/2} - 1 \right) + 1 
\asymp \left( d - 1 \right)^{t/2}.
\]  
Then 
\[
\E \left[ R_{t} \right] 
\leq 
\left( 1 - \gamma \right) \frac{t}{2} + \frac{\log \left( C t \right)}{\log \left( d -1 \right)} + 2.
\]
\item For every $\gamma \in \left( 0, 1 \right)$ there exists an adaptive diffusion protocol that satisfies~\eqref{eq:poly_obfuscation} with $C = 2(d-1)$ and also 
\[
\E \left[ R_{t} \right] \geq \left( 1 - \gamma \right) \frac{t}{2}
\]
for all even $t > 2 / \gamma$ (and for all even $t \leq 2 / \gamma$ we have $\E \left[ R_{t} \right] = t / 2  - 1$). 
\end{enumerate}
\end{theorem}

In particular, we see from Theorem~\ref{thm:local_spreading} that the power $\gamma$ in polynomial obfuscation (see~\eqref{eq:poly_obfuscation}) and the speed $\left( 1 - \gamma \right) / 2$ of local spreading are directly related. 
This precisely quantifies the tradeoff between local spreading and source obfuscation guarantees: the faster local spreading is---that is, the smaller $\gamma$ is---the weaker the source obfuscation guarantee.

\subsection{Organization}

The rest of the paper is organized as follows. We first prove Theorem~\ref{thm:multi_obs} in Section~\ref{sec:multi_obs}, since the proof relies only on a simple symmetry property of adaptive diffusion protocols on $\mathbb{T}_{d}$ and provides good intuition for the subsequent proofs. 
We then prove Theorem~\ref{thm:two_obs} in Section~\ref{sec:two_obs}; 
the proof of part~\eqref{thm:two_obs_detect} is similar to the proof of Theorem~\ref{thm:multi_obs} in Section~\ref{sec:multi_obs}, 
while the proof of part~\eqref{thm:two_obs_obfu} requires understanding the maximum likelihood estimator in the case of two observations. (Some cases in the proof of part~\eqref{thm:two_obs_obfu} of Theorem~\ref{thm:two_obs} are deferred to Appendix~\ref{sec:two_obs_obfu_odd}.) 
In Section~\ref{sec:local_spreading_proofs} we turn to studying local spreading and prove Theorem~\ref{thm:local_spreading}. 
Finally, we conclude in Section~\ref{sec:discussion} by discussing some implications and limitations of our results, how they relate to other works, as well as further questions for future research.

\section{Adaptive diffusion with $k \geq 3$ independent observations} \label{sec:multi_obs} 

In this section we prove Theorem~\ref{thm:multi_obs}. 
The main idea is simple and relies on a symmetry property of adaptive diffusion protocols on $\mathbb{T}_{d}$: that they send the virtual source in a uniformly random direction. 
First, note that if we remove the source $v^{*}$ from the tree $\mathbb{T}_{d}$ then it breaks into $d$ subtrees. 
The main observation is that the virtual source of an adaptive diffusion is equally likely to be in each subtree. 
This symmetry property alone guarantees a constant probability of detection when there are at least three independent observations, as we now explain.

Assume for now that $t_{1}, \ldots, t_{k}$ are even; 
the proof is cleaner in this case, though not much changes in the general case.
Recall that for an adaptive diffusion protocol the infected tree $G_{t}$ is a ball with center $\vs_{t}$ when $t$ is even. Hence from the infected tree $G_{t}$ we may determine the virtual source $\vs_{t}$. We may thus assume that the adversary is given $k$ independent virtual sources $\vs^{1}, \vs^{2}, \ldots, \vs^{k}$ (the time stamps of the virtual sources are not relevant for what follows). 
The main observation is that if $\vs^{1}$, $\vs^{2}$, and $\vs^{3}$ are in different subtrees, then $v^{*}$ is the unique vertex at the intersection of the three shortest paths connecting $\vs^{1}$ and $\vs^{2}$,  $\vs^{1}$ and $\vs^{3}$, and $\vs^{2}$ and~$\vs^{3}$; see Figure~\ref{fig:three_obs} for an illustration.

\begin{figure}[h!]
\centering
\includegraphics[scale=0.8]{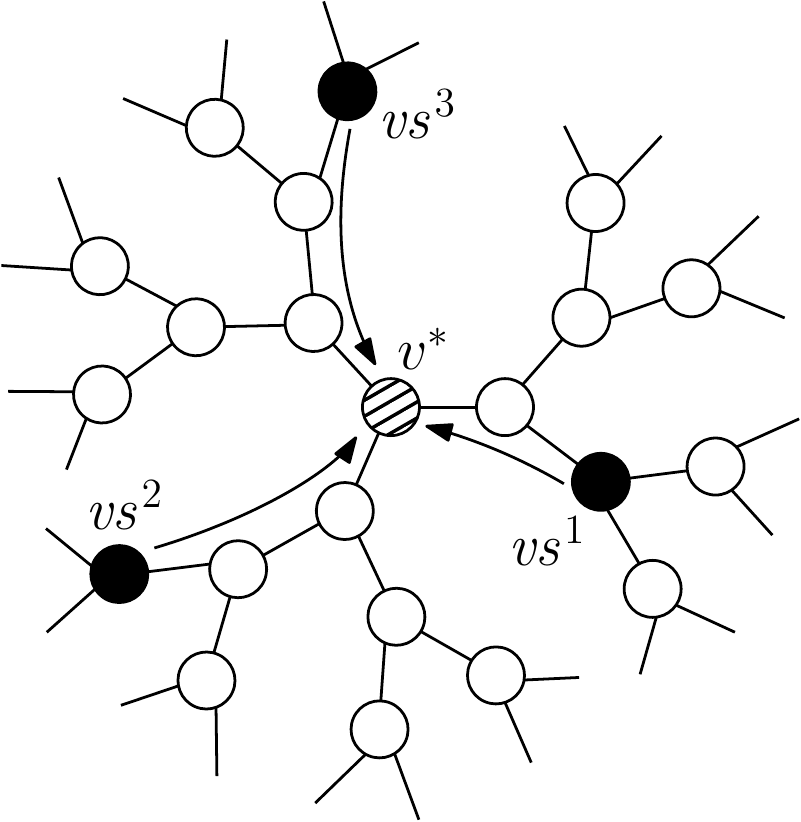}
\caption{\textbf{Detecting the source from three observations.} If the three virtual sources are in different subtrees, then the paths connecting them intersect in a single node: the source~$v^{*}$.}
\label{fig:three_obs}
\end{figure}

This immediately leads to a source detection algorithm: 
if the three shortest paths connecting $\vs^{1}$ and $\vs^{2}$,  $\vs^{1}$ and $\vs^{3}$, and $\vs^{2}$ and $\vs^{3}$ intersect at a single vertex, the algorithm outputs this vertex; 
if 
not, 
pick a vertex from the intersection uniformly at random. 
Since each virtual source is equally likely to be in each subtree, there is a constant probability that $\vs^{1}$, $\vs^{2}$, and $\vs^{3}$ are in different subtrees and therefore the algorithm successfully detects the source.

The proof that follows makes this formal and also presents an improved algorithm when the number of observations $k$ is large, in order to show that the detection probability goes to $1$ as $k \to \infty$.

\begin{proof}[Proof of Theorem~\ref{thm:multi_obs}] 
We start with some notational preliminaries. 
For distinct nodes $x, y \in \mathbb{T}_{d}$, let $T_{x}^{y}$ denote the subtree of $\mathbb{T}_{d}$ away from $y$ in the direction of $x$. In other words, if $y$ were removed from $\mathbb{T}_{d}$ then the tree would break into a forest of $d$ trees and $T_{x}^{y}$ is the tree that contains $x$. Formally, if $n_{x}^{y}$ is the neighbor of $y$ that is closest to $x$, then 
\[
T_{x}^{y} := \left\{ z \in \mathbb{T}_{d} : \delta \left( z, n_{x}^{y} \right) < \delta \left( z, y \right) \right\}.
\]

We first assume, for simplicity, that $t_{1}, \ldots, t_{k}$ are all even; this simplifies the proof and we explain at the end what changes if some of these times are odd. 
Then for every $i \in \left\{ 1, \ldots, k \right\}$ we have that $G_{t_{i}}^{i}$ is a ball (of radius $t_{i} / 2$) with center $\vs_{t_{i}}^{i}$, the virtual source at time $t_{i}$. 
Thus we may assume that the adversary observes $k \geq 3$ independent virtual sources $\vs^1, \ldots, \vs^k \in \mathbb{T}_d$; as the time indices do not play a role in what follows, we drop them for notational convenience. 
We first define an estimator $\wh{v}$ using only the first three samples ($\vs^{1}$, $\vs^{2}$, and $\vs^{3}$) and show that it detects the source with constant probability. 
For $i,j \in \{1,2,3\}$ let $P_{ij}$ denote the set of vertices in the unique path in $\mathbb{T}_{d}$ between $\vs^{i}$ and $\vs^{j}$. 
If the three paths $P_{12}$, $P_{13}$, and $P_{23}$ intersect in a single vertex, let $\wh{v}$ be this vertex. 
If the intersection of $P_{12}$, $P_{13}$, and $P_{23}$ contains more than one vertex, let $\wh{v}$ pick a vertex from this intersection uniformly at random. 

Consider the event $A$ where the three virtual sources take different first steps away from the source. 
By the construction of adaptive diffusion, this is the same as the virtual sources being in different subtrees for all positive times; 
that is, 
\[
A = \left\{ T_{\vs^{1}}^{v^{*}} \cap T_{\vs^{2}}^{v^{*}} = \emptyset \right\} 
\cap \left\{ T_{\vs^{1}}^{v^{*}} \cap T_{\vs^{3}}^{v^{*}} = \emptyset \right\} 
\cap \left\{ T_{\vs^{2}}^{v^{*}} \cap T_{\vs^{3}}^{v^{*}} = \emptyset \right\}.
\]
On the event $A$ we have that 
$P_{12} \cap P_{13} \cap P_{23} = \left\{ v^{*} \right\}$ 
and, hence, $\wh{v} = v^{*}$. 
That is, on the event $A$, 
the estimator correctly detects the source of the diffusion. 
Since the direction of the first step of a virtual source is uniformly random among the $d$ choices 
and the different samples are independent, 
we have that 
$\p \left( A \right) = \frac{\left( d - 1 \right) \left( d - 2 \right)}{d^{2}}$, 
which concludes this part of the proof.

We now explain how more samples can be used to achieve a detection probability that converges to 1 exponentially in $k$ as $k \to \infty$. 
For any vertex $v \in \mathbb{T}_{d}$ and $w$ a neighbor of $v$, define 
\[
N_{w} \left( v \right) := \# \left\{ j \in \left[ k \right] : \vs^{j} \in T_{w}^{v} \right\}.
\]
That is, $N_{w} \left( v \right)$ counts the number of virtual sources in the subtree of $\mathbb{T}_{d}$ away from $v$ in the direction of $w$. 
Using these quantities we define the following estimator: 
\begin{equation} \label{multi_est}
\wh{w} := \argmin_{v \in \mathbb{T}_{d}} \max_{w: (w,v) \in E} N_{w} \left( v \right), 
\end{equation}
provided that this is well-defined (i.e., the minimum is attained at a single vertex); 
if this is not well-defined, let $\wh{w}$ be an arbitrary vertex. 
Let $w_{1}, \ldots, w_{d}$ denote the neighbors of $v^{*}$ in $\mathbb{T}_{d}$ 
and let $Y := \left( N_{w_{1}} \left( v^{*} \right), \ldots, N_{w_{d}} \left( v^{*} \right) \right)$. 
We now argue that if $\left\| Y \right\|_{\infty} < k/2$, then $\wh{w} = v^{*}$, that is, the estimator correctly detects the source of the diffusion. 

First, observe that $\max_{w : (w, v^{*}) \in E} N_{w} \left( v^{*} \right) = \left\| Y \right\|_{\infty}$, which is less than $k/2$ under the assumption. 
Second, if $v \neq v^{*}$, then there must exist $w'$ a neighbor of $v$ and $i \in \left[ d \right]$ such that 
\[
T_{w'}^{v} \supseteq \bigcup_{j \in \left[ d \right] \setminus \left\{ i \right\}} T_{w_{j}}^{v^{*}}.
\]
This implies that 
\[
N_{w'} \left( v \right) \geq \sum_{j \in \left[ d \right] \setminus \left\{ i \right\}} N_{w_{j}} \left( v^{*} \right) 
= k - N_{w_{i}} \left( v^{*} \right) 
\geq k - \left\| Y \right\|_{\infty} > k/2,
\]
where we used that $\left\| Y \right\|_{1} = k$, as well as the assumption that $\left\| Y \right\|_{\infty} < k/2$. 
Consequently, 
$\max_{w : (w,v)\in E} N_{w} \left( v \right) 
\geq N_{w'} \left( v \right) > k/2$ 
and, hence, $\wh{w} \neq v$. 
We have thus shown that $\left\| Y \right\|_{\infty} < k/2$ implies that $\wh{w} = v^{*}$. 

To conclude, we estimate from below the probability that $\left\| Y \right\|_{\infty} < k/2$, or rather, we estimate from above the complimentary event that $\left\| Y \right\|_{\infty} \geq k/2$. 
First, by a union bound and symmetry we have that 
$\p \left( \left\| Y \right\|_{\infty} \geq k/2 \right) 
\leq d \times \p \left( N_{w_{1}} \left( v^{*} \right) \geq k/2 \right)$. 
Now since 
$N_{w_{1}} \left( v^{*} \right) \sim \Bin \left( k, 1/d \right)$, 
we have by a Chernoff bound that 
\[
\p \left( N_{w_{1}} \left( v^{*} \right) \geq k/2 \right) 
= \p \left( N_{w_{1}} \left( v^{*} \right) - \E \left[ N_{w_{1}} \left( v^{*} \right) \right] \geq \tfrac{d-2}{2d} k \right) 
\leq \exp \left( - \tfrac{\left( d-2 \right)^{2}}{2d^{2}} k \right).
\]

Finally, we return to our simplifying assumption that the observation times $t_{1}, \ldots, t_{k}$ are all even. If $t_{i}$ is odd, then there are two cases. If $G_{t_{i}}^{i}$ is a ball, then it is a ball with center $\vs_{t_{i}}^{i}$, so the adversary can again determine the virtual source at time $t_{i}$ and everything is unchanged. If $G_{t_{i}}^{i}$ is not a ball, then it is symmetric about the edge connecting $\vs_{t_{i}-1}^{i}$ and $\vs_{t_{i}}^{i}$. Thus the adversary can determine the set $\left\{ \vs_{t_{i}-1}^{i}, \vs_{t_{i}}^{i} \right\}$. Picking either element of the set as the virtual source, the remainder of the proof goes through unchanged.
\end{proof}

At first glance, it may appear that computing the estimator $\wh{w}$ requires solving a minimization problem over the entire infinite tree $\mathbb{T}_d$, but this is not the case. For every vertex $v$ that is not on a shortest path between two virtual sources we have that $\max_{w: (w,v)\in E} N_{w} (v) = k$ and therefore $\wh{w}$ must lie on a shortest path between two virtual sources. 
Moreover, the distance between any two virtual sources is at most $2\max_{i\in [k]} t_{i}$. 
Thus the minimization problem in~\eqref{multi_est} is over a set of size $O(k^{2} \max_{i\in [k]} t_{i})$ and can be solved efficiently.

\section{Adaptive diffusion with two independent observations} \label{sec:two_obs} 

In this section we prove Theorem~\ref{thm:two_obs}. 
We start with the proof of part~(\ref{thm:two_obs_detect}) in Section~\ref{sec:two_obs_detect}, 
which builds on similar ideas as the proof of Theorem~\ref{thm:multi_obs} in Section~\ref{sec:multi_obs}.
Then, in order to prove part~(\ref{thm:two_obs_obfu}) of Theorem~\ref{thm:two_obs}, 
we need to understand the maximum likelihood estimator---this is done in Section~\ref{sec:MLE}. 
Due to the nature of adaptive diffusion, we have to deal with even and odd times separately. 
To focus on the key insights and computations, we first prove Theorem~\ref{thm:two_obs}(\ref{thm:two_obs_obfu}) when $t_{1}$ and $t_{2}$ are both even---this is in Section~\ref{sec:two_obs_obfu_even}. 
The cases when one or both of $t_{1}$ and $t_{2}$ are odd are similar but more complicated, while not adding anything conceptually---hence we defer the proof in these cases to Appendix~\ref{sec:two_obs_obfu_odd}.

\subsection{Source detection}\label{sec:two_obs_detect}

The proof of Theorem~\ref{thm:two_obs}(\ref{thm:two_obs_detect}) builds on similar ideas as the proof of Theorem~\ref{thm:multi_obs} in Section~\ref{sec:multi_obs}. Recall the notation that we introduced in Section~\ref{sec:multi_obs}, which we use here. 

\begin{proof}[Proof of Theorem~\ref{thm:two_obs}(\ref{thm:two_obs_detect})]
 Assume first that $t_{1}$ and $t_{2}$ are even; this simplifies the proof and we explain at the end what changes if either time is odd. 
 Then for $i \in \left\{ 1, 2 \right\}$ we have that $G_{t_{i}}^{i}$ is a ball of radius $t_{i}/2$ with center $\vs_{t_{i}}^{i}$. 
 The adversary can thus determine the two virtual sources $\vs^{1} \equiv \vs_{t_{1}}^{1}$ and~$\vs^{2} \equiv \vs_{t_{2}}^{2}$. 
 
 By definition we always have that $v^{*} \in V_{t_{1}}^{1} \cap V_{t_{2}}^{2}$, that is, the source $v^{*}$ is contained in both sets of infected nodes. 
 Let $P_{12}$ denote the set of vertices that are on the path in $\mathbb{T}_{d}$ between $\vs^{1}$ and~$\vs^{2}$, excluding $\vs^{1}$ and $\vs^{2}$. Furthermore, define the set 
 $S := P_{12} \cap V_{t_{1}}^{1} \cap V_{t_{2}}^{2}$.  
 Let $A_{12}$ denote the event that $\vs^{1}$ and $\vs^{2}$ are in different subtrees away from $v^{*}$; that is, 
 \begin{equation}\label{eq:A12}
  A_{12} := \left\{ T^{v^{*}}_{\vs^{1}} \cap T^{v^{*}}_{\vs^{2}} = \emptyset \right\}.
 \end{equation}
 Since the two diffusions are independent and the first step of the virtual source is to a uniformly random neighbor of $v^{*}$, we have that $\p \left( A_{12} \right) = (d-1)/d$. The main observation is that, on the event $A_{12}$, we have that $v^{*} \in P_{12}$; see Figure~\ref{fig:two_obs} for an illustration.\footnote{The two virtual sources can indeed be excluded from $P_{12}$ and we still have that $v^{*} \in P_{12}$ on the event~$A_{12}$. This is because the virtual source can never be the true source, by construction. This assumes that $t_{1}, t_{2} \geq 1$---which holds, since we assume in the proof that $t_{1}, t_{2} \geq 2$. In any case, if $\min \left\{ t_{1}, t_{2} \right\} < 2$, then one of the observed snapshots contains at most two vertices, so a random guess succeeds in identifying the source with probability at least $1/2$.} Consequently, on the event $A_{12}$ we also have that~$v^{*} \in S$. 
 \begin{figure}[h!]
\centering
\includegraphics[scale=0.6]{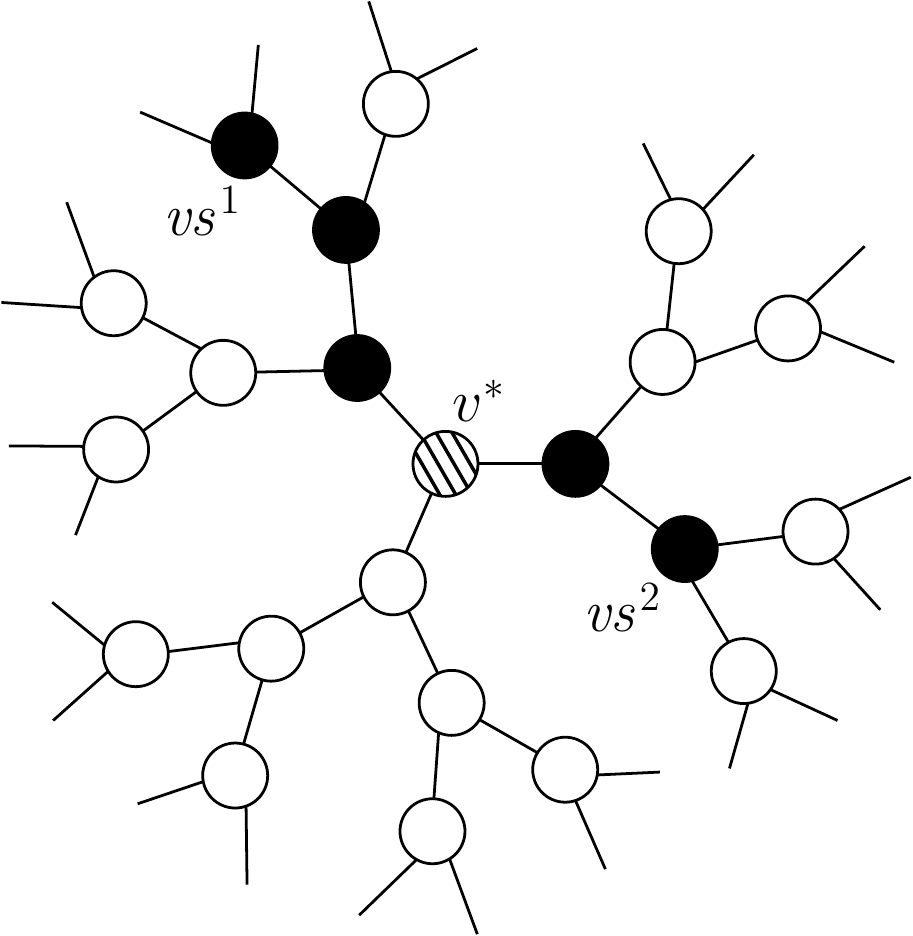}
\caption{\textbf{Detecting the source from two observations.} If the two virtual sources are in different subtrees, then the path connecting them contains the source~$v^{*}$.}
\label{fig:two_obs}
\end{figure}

 This suggests a natural estimator: if $S \neq \emptyset$, let $\widehat{v}$ be a uniformly randomly chosen node from $S$ (note that $S$ is a measurable function of $G_{t_{1}}^{1}$ and $G_{t_{2}}^{2}$); if $S = \emptyset$ (this occurs when $\delta( \vs^{1}, \vs^{2}) \leq 1$), let $\wh{v}$ be an arbitrary node. Then, given $A_{12}$ and $S$, the conditional probability that $\widehat{v} = v^{*}$ is $1/|S|$ (note that $A_{12}$ implies that $|S| \geq 1$, as we argued above). 
 We have thus shown that 
 \[
  \p \left( \wh{v} = v^{*} \right) 
  \geq 
  \p \left( \wh{v} = v^{*} \, \middle| \, A_{12} \right) \p \left( A_{12} \right) 
  = \E \left[ 1/|S| \, \middle| \, A_{12} \right] \frac{d-1}{d}.
 \]
 To conclude, it suffices to show that $\left|S\right| \leq \min \left\{ t_{1} / 2, t_{2} / 2 \right\}$ whenever $A_{12}$ holds. 
 To see this, note that the intersection $P_{12} \cap V_{t_{1}}^{1}$ contains at most $t_{1} / 2$ nodes, since $G_{t_{1}}^{1}$ is a (closed) ball of radius $t_{1} / 2$ centered at $\vs^{1}$, the path $P_{12}$ starts at the virtual source $\vs^{1}$, and $\vs^{1}$ is not included in $P_{12}$. 
 Thus $|S| \leq \left| P_{12} \cap V_{t_{1}}^{1} \right| \leq t_{1} / 2$. 
 Similarly, $P_{12} \cap V_{t_{2}}^{2}$ contains at most $t_{2} / 2$ nodes, and the claim follows.

 Finally, we explain what changes when $t_{i}$ is odd for $i = 1$ and/or $i=2$. 
 If $G_{t_{i}}^{i}$ is a ball, then its center is $\vs_{t_{i}}^{i}$, so the adversary can again determine the virtual source at time $t_{i}$ and everything is unchanged. 
 If $G_{t_{i}}^{i}$ is not a ball, then it is symmetric about the edge connecting $\vs_{t_{i}-1}$ and $\vs_{t_{i}}$. 
 Thus the adversary can determine the set $\left\{ \vs_{t_{i} - 1}, \vs_{t_{i}} \right\}$. Connecting \emph{both} of these virtual sources with the other virtual source(s), we again obtain a path, where now at both ends of the path we have either one or two virtual sources. In any case, we can define $P_{12}$ analogously, where again the known virtual sources are not considered as part of $P_{12}$. The rest of the proof is unchanged. 
\end{proof}

\subsection{Maximum likelihood source estimation}\label{sec:MLE} 

In order to prove Theorem~\ref{thm:two_obs}\eqref{thm:two_obs_obfu}, 
we need to understand maximum likelihood source estimation. 
Here 
we discuss this 
for adaptive diffusions in general. 
Recall that an adaptive diffusion protocol is given by the probabilities 
$\alpha(t,h) \in [0,1]$, 
with $t \in \left\{ 2, 4, 6, \ldots \right\}$ 
and $h \in \left\{ 1, 2, 3, \ldots, t/2 \right\}$, 
which determine the distribution of the path of the virtual source 
$\left\{ \vs_{t} \right\}_{t \geq 0}$. 
Let $h_{t} := \delta \left( \vs_{t}, v^{*} \right)$ denote the graph distance between $\vs_{t}$ and $v^{*}$, and let $p(t,h) := \p \left( h_{t} = h \right)$ denote the distribution of $h_{t}$.

When determining the likelihood function 
$L(v) = \p \left( G_{t} \, \middle| \, v^{*} = v \right)$ 
we have to specify whether the value of $t$ is known or not 
(since it is not always possible to infer the value of $t$ from the observation~$G_{t}$). 
We assume in the following that $t$ is \emph{known}. 
Note that knowing $t$ can only help the adversary 
and hence any upper bounds on the success probability of the MLE under this assumption 
still hold without this assumption. 
Furthermore, the rumor source detection results (Theorem~\ref{thm:two_obs}\eqref{thm:two_obs_detect} and Theorem~\ref{thm:multi_obs}) hold regardless of whether we assume this or not. 
Finally, note that this assumption is also what is used in previous works~\cite{fanti2015spy,fanti2016irregular,fanti2017hide}.

We now determine the likelihood function 
$L(v) = \p \left( G_{t} \, \middle| \, v^{*} = v \right)$ 
for even $t$; 
it is similar for odd~$t$, but we leave this for later. 
First, we always have that $v^{*} \in V_{t} \setminus \left\{ \vs_{t} \right\}$, so $L(v) = 0$ if $v \notin V_{t} \setminus \left\{ \vs_{t} \right\}$. 
Next, since $G_{t}$ is a ball of radius~$t/2$ with center $\vs_{t}$, 
it is fully determined by the position of the virtual source, together with the time~$t$. 
It is important to note a key symmetry property of adaptive diffusion: 
all nodes at a particular distance from the virtual source are equally likely to have been the source. 
This is because the virtual source always moves to a uniformly randomly chosen neighbor away from the source. 
Thus the distribution of the virtual source is completely determined by the distribution of~$h_{t}$. 
Altogether, since there are $d(d-1)^{h-1}$ nodes at distance $h \geq 1$ from a particular vertex, we obtain that 
\begin{equation}\label{eq:likelihood_even}
 L(v) =  \frac{1}{d(d-1)^{\delta \left( v, \vs_{t} \right)-1}} p \left( t , \delta \left( v, \vs_{t} \right) \right) \mathbf{1}_{\left\{ v \in V_{t} \setminus \left\{ \vs_{t} \right\} \right\}}.
\end{equation}

Now assume that we have $k$ independent observations of infected subgraphs, $G_{t_{i}}^{i} = \left( V_{t_{i}}^{i}, E_{t_{i}}^{i} \right)$ for $i \in \left\{ 1, \ldots, k \right\}$, started from a fixed source $v^{*}$. 
Assume also, for now, that all the times $t_{1}, \ldots, t_{k}$ are even. 
Then, by independence, the likelihood function is 
\[
 L(v) = \left( \frac{d-1}{d} \right)^{k} \prod_{i=1}^{k} p \left(t_{i}, X_{i}(v) \right) \cdot \left( d-1 \right)^{-X_{i}(v)}  
 \mathbf{1}_{\left\{ v \in \bigcap_{i=1}^{k} \left( V_{t_{i}}^{i} \setminus \left\{ \vs_{t_{i}}^{i} \right\} \right) \right\}},
\]
where we have introduced 
\begin{equation}\label{eq:Xi}
 X_{i} (v) := \delta \left( v, \vs_{t_{i}}^{i} \right)
\end{equation}
for convenience (and recall that we can determine $\vs_{t_{i}}^{i}$, and thus also $X_{i}(v)$, from $G_{t_{i}}^{i}$). 
By taking logarithms, we obtain that the MLE satisfies 
\begin{equation}\label{eq:mle}
 \wh{v}_{\ML} \in \argmax_{v \in \bigcap_{i=1}^{k} \left( V_{t_{i}}^{i} \setminus \left\{ \vs_{t_{i}}^{i} \right\} \right)} 
 \sum_{i=1}^{k} \left\{ \log p \left( t_{i}, X_{i}(v) \right) - X_{i}(v) \log(d-1) \right\}.
\end{equation}

We now turn to determining the likelihood function 
$L(v) = \p \left( G_{t} \, \middle| \, v^{*} = v \right)$ 
for odd~$t$. 
This is similar to the case of even~$t$, but there are slight differences. 
Specifically, there are two cases to distinguish: 
when $t$ is odd, 
the observed graph $G_{t}$ is either a ball or it is not (in which case it consists of two balanced rooted trees of depth $(t-1)/2$, whose roots are connected by an edge). 

The former case occurs when the virtual source does not move at time $t-1$, that is, when 
$\vs_{t-1} = \vs_{t}$. 
In this case, we know that $G_{t-1} = G_{t}$, 
we know the likelihood of $G_{t-1}$ (which is given by~\eqref{eq:likelihood_even} with $t$ replaced by $t-1$), 
and in order to obtain the likelihood of $G_{t}$ we have to multiply this by the probability that $\vs_{t-1} = \vs_{t}$, 
which is $\alpha(t-1,X(v))$, where $X(v) = \delta(v,\vs_{t-1}) = \delta(v,\vs_{t})$. 

In the latter case, when $G_{t}$ is not a ball, we know that the virtual source moved at time $t-1$. 
Furthermore, we can determine the set $\left\{ \vs_{t-1}, \vs_{t} \right\}$, as these two vertices are connected by the central edge of $G_{t}$. 
In this case, we define $X(v) := \min \left\{ \delta(v,\vs_{t-1}), \delta(v,\vs_{t}) \right\}$ (note that $X(v)$ can be determined from $G_{t}$). 
In order to obtain the likelihood of $G_{t}$ we have to multiply the expression in~\eqref{eq:likelihood_even} (with $t$ replaced by $t-1$ and $\delta(v,\vs_{t})$ replaced with $\min \left\{ \delta(v,\vs_{t-1}), \delta(v,\vs_{t}) \right\}$) 
with the probability that $\vs_{t-1} \neq \vs_{t}$, 
which is $1 - \alpha(t-1,X(v))$. 

Altogether, when $t$ is odd we have that the likelihood function is 
\begin{equation}\label{eq:likelihood_odd}
 L(v) = 
 \begin{cases}
  \frac{1}{d(d-1)^{X(v)-1}} p \left( t - 1 , X(v) \right) \alpha(t-1,X(v)) \mathbf{1}_{\left\{ v \in V_{t} \setminus \left\{ \vs_{t} \right\} \right\}} 
  &\text{ if } G_{t} \text{ is a ball}, \\
  \frac{1}{d(d-1)^{X(v)-1}} p \left( t - 1 , X(v) \right) \left\{ 1 - \alpha(t-1,X(v)) \right\} \mathbf{1}_{\left\{ v \in V_{t} \setminus \left\{ \vs_{t-1}, \vs_{t} \right\} \right\}} 
  &\text{ otherwise},
 \end{cases}
\end{equation}
where $X(v) := \min \left\{ \delta(v,\vs_{t-1}), \delta(v,\vs_{t}) \right\}$ (note that this definition of $X(v)$ works for both cases; when $G_{t}$ is a ball then $\vs_{t-1} = \vs_{t}$ and hence $X(v) = \delta(v,\vs_{t-1}) = \delta(v,\vs_{t})$).

\subsection{Source obfuscation --- even times}\label{sec:two_obs_obfu_even}

We are now ready to prove Theorem~\ref{thm:two_obs}(\ref{thm:two_obs_obfu}). 
We first prove this when both $t_{1}$ and $t_{2}$ are even. 
This is done in order to highlight the key insights and computations. 
The remaining cases (when one or both of $t_{1}$ and $t_{2}$ are odd) 
are similar but more complicated and hence are deferred to Appendix~\ref{sec:two_obs_obfu_odd}.

\begin{proof}[Proof of Theorem~\ref{thm:two_obs}(\ref{thm:two_obs_obfu}) when $t_{1}$ and $t_{2}$ are both even.] 
 We may assume in the following that $t_{1}, t_{2} \geq 4$, since when $\min \left\{ t_{1}, t_{2} \right\} = 2$ then the right hand side of~\eqref{eq:two_obs_obfu} is greater than $1$ and thus the statement is vacuously true. 

 Consider the adaptive diffusion protocol---which we term the \emph{uniform protocol} $\mathcal{U}$ for reasons to become clear---given by the probabilities 
 \begin{equation}\label{eq:alpha_uniform}
  \alpha_{\mathcal{U}} (t,h) := \frac{t-2h+2}{t+2}
 \end{equation}
 for $t \in \left\{ 2, 4, 6, \ldots \right\}$ and $h \in \left\{ 1, 2, \ldots, t/2 \right\}$. 
 This is the same protocol introduced by Fanti~et~al.~\cite{fanti2015spy} to achieve perfect obfuscation from a single snapshot on $\mathbb{Z}$---the difference is that here we use this protocol regardless of the degree $d$. 
 The important property of this protocol is that the distance 
 $h_{t} := \delta \left( \vs_{t}, v^{*} \right)$ between the virtual source $\vs_{t}$ and the true source $v^{*}$ is \emph{uniformly distributed} over the set of possible values $\left\{ 1, 2, \ldots, t/2 \right\}$, for all even $t$. 
 That is, for all even $t$ we have that 
 \begin{equation}\label{eq:p_uniform}
  p_{\mathcal{U}} (t,h) = \frac{2}{t} \mathbf{1}_{\left\{ h \in \left\{ 1, 2, \ldots, t/2 \right\} \right\}}.
 \end{equation}
 This can be shown by induction; we leave the details to the reader. 
 
 We now turn to analyzing the maximum likelihood estimator of the source, $\wh{v}_{\ML}$, given two independent snapshots $G_{t_{1}}^{1}$ and $G_{t_{2}}^{2}$. 
 Recall that we assume now that $t_{1}$ and $t_{2}$ are both even. 
 The adversary can thus determine the two virtual sources 
 $\vs^{1} \equiv \vs_{t_{1}}^{1}$ and $\vs^{2} \equiv \vs_{t_{2}}^{2}$. 
 By plugging in~\eqref{eq:p_uniform} into~\eqref{eq:mle}, 
 we obtain that the MLE satisfies 
 \[
  \wh{v}_{\ML} \in \argmin_{v \in V_{t_{1}}^{1} \cap V_{t_{2}}^{2} \setminus \left\{ \vs^{1}, \vs^{2} \right\}} 
  \left( X_{1}(v) + X_{2}(v) \right), 
 \]
 where recall from~\eqref{eq:Xi} that 
 $X_{i} (v) = \delta(v,\vs^{i})$ 
 for $i \in \left\{ 1, 2 \right\}$. 
 In words, 
 the maximum likelihood estimator minimizes the sum of the distances to the two virtual sources, over all nodes that were infected in both diffusions, excluding the two virtual sources. 
 
 To understand the MLE better we distinguish three cases, the last one being the most important: 
 \begin{enumerate}[(1)]
  \item\label{case:MLE_even_even_1} If $\vs^{1} = \vs^{2}$, then $\wh{v}_{\ML}$ chooses a neighbor of $\vs^{1} = \vs^{2}$ uniformly at random. 
  \item\label{case:MLE_even_even_2} If $\delta( \vs^{1}, \vs^{2}) = 1$, then $\wh{v}_{\ML}$ chooses a neighbor of the set $\left\{ \vs^{1}, \vs^{2} \right\}$ uniformly at random.\footnote{Here we use that $t_{1},t_{2} \geq 4$, to ensure that all neighbors of the set $\left\{ \vs^{1}, \vs^{2} \right\}$ are in $V_{t_{1}}^{1} \cap V_{t_{2}}^{2}$.}
  \item If $\delta( \vs^{1}, \vs^{2}) \geq 2$, then $X_{1}(v) + X_{2}(v)$ is minimized when $v$ is on the shortest path between $\vs^{1}$ and $\vs^{2}$. 
 Let $P_{12}$ denote the set of vertices that are on the shortest path between $\vs^{1}$ and $\vs^{2}$, excluding $\vs^{1}$ and $\vs^{2}$. 
 Furthermore, define the set $S := P_{12} \cap V_{t_{1}}^{1} \cap V_{t_{2}}^{2}$ 
 and note that when $\delta( \vs^{1}, \vs^{2}) \geq 2$, then $S$ is nonempty, because the vertex in $P_{12}$ that is closest to $v^{*}$ is always in~$S$. 
 We have thus argued that the likelihood function is maximized at the nodes in $S$ and 
 thus the maximum likelihood estimator $\wh{v}_{\ML}$ chooses a node from $S$ uniformly at random. 
 \end{enumerate}
 Note that $\wh{v}_{\ML}$ is (essentially) the same as the estimator $\wh{v}$ introduced in the proof of part~(\ref{thm:two_obs_detect}) of Theorem~\ref{thm:two_obs}. 
 
 Let $A_{12}$ denote the event that $\vs^{1}$ and $\vs^{2}$ are in different subtrees away from $v^{*}$ (see~\eqref{eq:A12}),  
 and note that $\p \left( A_{12} \right) = (d-1)/d$. 
 Observe that if the event $A_{12}$ holds, then necessarily $\delta( \vs^{1}, \vs^{2} ) \geq 2$, 
 and hence the first two cases above imply that $A_{12}$ does not hold. 
 To compute the probability that the MLE $\wh{v}_{\ML}$ is correct, 
 we may condition on whether or not $A_{12}$ holds: 
 \begin{align}
  \p \left( \wh{v}_{\ML} = v^{*} \right) 
 &= \p \left( \wh{v}_{\ML} = v^{*} \, \middle| \, A_{12} \right) \p \left( A_{12} \right) 
 + \p \left( \wh{v}_{\ML} = v^{*} \, \middle| \, A_{12}^{C} \right) \p \left( A_{12}^{C} \right) \notag\\
 &= \p \left( \wh{v}_{\ML} = v^{*} \, \middle| \, A_{12} \right) \cdot \frac{d-1}{d}
 + \p \left( \wh{v}_{\ML} = v^{*} \, \middle| \, A_{12}^{C} \right) \cdot \frac{1}{d}. \label{eq:MLE_total_prob}
 \end{align}

 Let us now turn to computing $\p \left( \wh{v}_{\ML} = v^{*} \, \middle| \, A_{12}^{C} \right)$. 
 There are two cases when the MLE can be correct, given that $A_{12}$ does not hold. 
 First, corresponding to Case~\eqref{case:MLE_even_even_1} above: 
 if $\vs^{1} = \vs^{2}$ and $\delta(v^{*},\vs^{1}) = \delta(v^{*},\vs^{2}) = 1$, 
 then the MLE is correct with probability $1/d$. 
 Second, corresponding to Case~\eqref{case:MLE_even_even_2} above: 
 if $\delta(v^{*},\vs^{1}) = \delta(\vs^{1},\vs^{2}) = 1$ 
 or if $\delta(v^{*},\vs^{2}) = \delta(\vs^{1},\vs^{2}) = 1$, 
 then the MLE is correct with probability $1/(2d-2)$. 
 If $\delta( \vs^{1}, \vs^{2} ) \geq 2$ and $A_{12}$ does not hold, 
 then $\wh{v}_{\ML} \neq v^{*}$. 
 Putting these together and using~\eqref{eq:p_uniform} we obtain that 
 \begin{equation}\label{eq:MLE_A12C}
  \p \left( \wh{v}_{\ML} = v^{*} \, \middle| \, A_{12}^{C} \right) 
  = \frac{2}{t_{1}} \cdot \frac{2}{t_{2}} \cdot \frac{1}{d} + 2 \cdot \frac{2}{t_{1}} \cdot \frac{2}{t_{2}} \cdot \frac{1}{2d-2} 
  = \frac{4}{t_{1} t_{2}} \left( \frac{1}{d} + \frac{1}{d-1} \right).
 \end{equation}

 We now turn to computing 
 $\p \left( \wh{v}_{\ML} = v^{*} \, \middle| \, A_{12} \right)$. 
 Given $A_{12}$ and $S$, the conditional probability that $\wh{v}_{\ML} = v^{*}$ is $1/|S|$. We thus have that 
 \begin{equation}\label{eq:exp1/S}
  \p \left( \wh{v}_{\ML} = v^{*} \, \middle| \, A_{12} \right) 
  = \E \left[ 1 / |S| \, \middle| \, A_{12} \right].
 \end{equation}
 On the event $A_{12}$, we can express $|S|$ as a function of $X_{1}(v^{*})$ and $X_{2}(v^{*})$ as follows. First, the set $S$ always contains $v^{*}$ when $A_{12}$ holds. Next, there are $X_{1}(v^{*}) - 1$ nodes on the path $P_{12}$ between $v^{*}$ and $\vs^{1}$. However, only the $t_{2}/2 - X_{2}(v^{*})$ nodes of these that are closest to $v^{*}$ are in $V_{t_{2}}^{2}$ as well. 
 Similarly, there are $X_{2}(v^{*}) - 1$ nodes on the path $P_{12}$ between $v^{*}$ and $\vs^{2}$, but only the $t_{1}/2 - X_{1}(v^{*})$ nodes of these that are closest to $v^{*}$ are in $V_{t_{1}}^{1}$ as well. 
 Altogether, on the event $A_{12}$ we have that 
 \begin{equation}\label{eq:S}
  |S| = 1 + \min \left\{ X_{1}(v^{*}) - 1, t_{2}/2 - X_{2}(v^{*}) \right\} + \min \left\{ X_{2}(v^{*}) - 1, t_{1}/2 - X_{1}(v^{*}) \right\}.
 \end{equation}
 Recall from~\eqref{eq:p_uniform} that $X_{i}(v^{*})$ is uniformly distributed on $\left\{ 1, 2, \ldots, t_{i} / 2 \right\}$, for $i \in \left\{ 1,2 \right\}$. Moreover, $X_{1}(v^{*})$ and $X_{2}(v^{*})$ are independent. Both of these statements hold conditioned on $A_{12}$. 
 Therefore, plugging in~\eqref{eq:S} into~\eqref{eq:exp1/S} and writing out the expectation we obtain that 
 \begin{equation}\label{eq:double_sum}
  \p \left( \wh{v}_{\ML} = v^{*} \, \middle| \, A_{12} \right) 
  = \frac{1}{st} \sum_{j=1}^{s} \sum_{\ell=1}^{t} \frac{1}{1+ \min \left\{ j-1,t-\ell \right\} + \min \left\{ \ell - 1, s - j \right\}},
 \end{equation}
 where we have introduced 
 $s := \min \left\{ t_{1}, t_{2} \right\} / 2$ 
 and 
 $t := \max \left\{ t_{1}, t_{2} \right\} / 2$ 
 in order to abbreviate notation.  
 With this notation, we can write $|S|$ from~\eqref{eq:S} more succintly by breaking things into three cases, as follows: 
 \begin{itemize}
  \item If $X_{1}(v^{*}) + X_{2}(v^{*}) \leq s + 1$, 
  then $|S| = X_{1}(v^{*}) + X_{2}(v^{*}) - 1$. 
  \item If $s + 1 < X_{1}(v^{*}) + X_{2}(v^{*}) \leq t + 1$, 
  then $|S| = s$. 
  \item If $t + 1 < X_{1}(v^{*}) + X_{2}(v^{*})$, 
  then $|S| = 1 + s + t - \left( X_{1}(v^{*}) + X_{2}(v^{*}) \right)$. 
 \end{itemize}
 Accordingly, we can break the sum in~\eqref{eq:double_sum} into three parts. 
 Let $\mathcal{I} := \left\{ (j,\ell) : 1 \leq j \leq s, 1 \leq \ell \leq t \right\}$ 
 denote the index set over which we take the sum in~\eqref{eq:double_sum}. 
 We can write it as the disjoint union 
 $\cI = \cI_{1} \cup \cI_{2} \cup \cI_{3}$, 
 where 
 $\cI_{1} := \left\{ (j,\ell) \in \cI : j + \ell \leq s+1 \right\}$, 
 $\cI_{2} := \left\{ (j,\ell) \in \cI : s + 1 < j + \ell \leq t + 1 \right\}$, 
 and 
 $\cI_{3} := \left\{ (j,\ell) \in \cI : t + 1 < j + \ell \right\}$. 
 We now consider the index sets $\cI_{1}$, $\cI_{2}$, and $\cI_{3}$ separately. 
 
 First, suppose that $m \in \left\{ 2, 3, \ldots, s + 1 \right\}$. 
 There are $m-1$ pairs of indices $(j,\ell) \in \cI_{1}$ such that $j + \ell = m$. For each such index pair, the fraction in~\eqref{eq:double_sum} is equal to $1/(m-1)$. 
 Since there are $s$ different values of $m$, the sum over the index set $\cI_{1}$ is equal to $s$.
 
 Next, observe that $\left| \cI_{2} \right| = s(t-s)$. For every $(j,\ell) \in \cI_{2}$, the fraction in~\eqref{eq:double_sum} is $1/s$. 
 Therefore the sum over the index set $\cI_{2}$ is equal to $s(t-s) / s = t-s$. 
 
 Finally, suppose that $m \in \left\{ t + 2, \ldots, t + s \right\}$. 
 There are $1 + s + t - m$ pairs of indices $(j,\ell) \in \cI_{3}$ such that $j + \ell = m$. For each such index pair, the fraction in~\eqref{eq:double_sum} is equal to $1/(1+s+t-m)$. 
 Since there are $s-1$ different values of $m$, the sum over the index set $\cI_{3}$ is equal to $s-1$. 
 
 Putting together the previous three paragraphs, we thus have that 
 \[
  \sum_{j=1}^{s} \sum_{\ell=1}^{t} \frac{1}{1+ \min \left\{ j-1,t-\ell \right\} + \min \left\{ \ell - 1, s - j \right\}} 
  = s + t - 1.
 \]
 Plugging this back into~\eqref{eq:double_sum}, and returning to the notation of $t_{1}$ and $t_{2}$, we obtain that 
 \begin{equation}\label{eq:MLE_A12}
  \p \left( \wh{v}_{\ML} = v^{*} \, \middle| \, A_{12} \right)
  = \frac{s+t-1}{st} 
  = \frac{2 t_{1} + 2 t_{2} - 4}{t_{1} t_{2}}.
 \end{equation}
 Putting together~\eqref{eq:MLE_total_prob},~\eqref{eq:MLE_A12C}, and~\eqref{eq:MLE_A12}, we have obtained that 
 \[
  \p \left( \wh{v}_{\ML} = v^{*} \right)
  = \frac{d-1}{d} \cdot \frac{2 t_{1} + 2 t_{2} - 4}{t_{1} t_{2}} + \frac{1}{d} \cdot \frac{4}{t_{1} t_{2}} \left( \frac{1}{d} + \frac{1}{d-1} \right) 
  < \frac{d-1}{d} \cdot \frac{2 t_{1} + 2 t_{2}}{t_{1} t_{2}},
 \]
 where we used that $1/d + 1/(d-1) < 1$. 
 Using that $2t_{1} + 2t_{2} \leq 4 \max \left\{ t_{1}, t_{2} \right\}$, 
 we obtain the bound in~\eqref{eq:two_obs_obfu}, 
 when $t_{1}$ and $t_{2}$ are both even. 
\end{proof}

\section{Local spreading vs.\ source obfuscation} \label{sec:local_spreading_proofs} 

In this section we prove Theorem~\ref{thm:local_spreading}. 
Recall the notation we introduced in previous sections, which we use here as well. 
In particular, $h_{t} := \delta \left( \vs_{t}, v^{*} \right)$ denotes the graph distance between $\vs_{t}$ and $v^{*}$ and $p(t,h) := \p \left( h_{t} = h \right)$. 
We will also use the elementary inequalities 
\begin{equation}\label{eq:Nt}
\left( d - 1 \right)^{t/2} 
\leq N_{t} 
\leq \frac{d}{d-2} \left( d - 1 \right)^{t/2}.
\end{equation}

\begin{proof}[Proof of Theorem~\ref{thm:local_spreading}]
Our starting observation is that, due to the definition of adaptive diffusion protocols, we have that  
\begin{equation}\label{eq:Rt_identity}
R_{t} = \frac{t}{2} - h_{t}.
\end{equation}
Thus in order to understand $R_{t}$ it is equivalent to understand $h_{t}$. 

We first turn to part (a) of the theorem. We described the likelihood function in Section~\ref{sec:MLE}, see~\eqref{eq:likelihood_even} in particular, from which it follows that 
\begin{equation}\label{eq:MLE_correct_prob}
\p \left( \wh{v}_{\ML} = v^{*} \right) 
= 
\max_{1 \leq h \leq t / 2} \frac{p(t,h)}{d\left( d - 1 \right)^{h - 1}}.
\end{equation}
The assumption~\eqref{eq:poly_obfuscation} thus implies that 
\begin{equation}\label{eq:pth_bound}
p(t,h) 
\leq 
\frac{C d \left( d - 1 \right)^{h - 1}}{N_{t}^{\gamma}} 
\leq 
C d \left( d - 1 \right)^{h - \gamma t / 2 - 1} 
\end{equation}
for all $1 \leq h \leq t / 2$, where in the second inequality we used~\eqref{eq:Nt}. 
Now define 
\[
m_{t} := \frac{\gamma t}{2} - \frac{\log \left( C t \right)}{\log \left( d - 1 \right)} - 1. 
\]
We then have that 
\begin{align*}
\p \left( h_{t} \leq m_{t} \right) 
&\leq \sum_{h = 1}^{\left\lfloor m_{t} \right\rfloor} C d \left( d - 1 \right)^{h - \gamma t / 2 - 1} 
= C d \left( d - 1 \right)^{-\gamma t / 2} \frac{\left( d - 1 \right)^{\left\lfloor m_{t} \right\rfloor} - 1}{d-2} \\
&\leq C d \left( d - 1 \right)^{m_{t} - \gamma t / 2} 
= \frac{d}{d-1} \cdot \frac{1}{t} \leq \frac{2}{t}.
\end{align*}
In particular, we thus have that 
$\p \left( h_{t} > m_{t} \right) \geq 1 - 2 / t$. 
Therefore 
\[
\E \left[ h_{t} \right] 
\geq m_{t} \p \left( h_{t} > m_{t} \right) 
\geq m_{t} \left(  1 - \frac{2}{t} \right)
\geq 
\frac{\gamma t}{2} - \frac{\log \left( C t \right)}{\log \left( d - 1 \right)} - 1 - \gamma.
\]
Now using~\eqref{eq:Rt_identity} we have that 
\[
\E \left[ R_{t} \right] 
= \frac{t}{2} - \E \left[ h_{t} \right] 
\leq \left( 1 - \gamma \right) \frac{t}{2} + \frac{\log \left( C t \right)}{\log \left( d - 1 \right)} + 1 + \gamma,
\]
which concludes the proof of part (a) of the theorem.

We now turn to part (b) of the theorem. Consider the adaptive diffusion protocol defined as follows: 
\begin{itemize}
\item For $t \leq 2 / \gamma$, let $\alpha \left( t, h \right) = 1$ for all $1 \leq h \leq t/2$. 
\item For $t > 2 / \gamma$, let 
$\alpha \left( t, h \right) = 1$ if 
$\left\lfloor \gamma t / 2 \right\rfloor = \left\lfloor \gamma \left( t / 2 + 1 \right) \right\rfloor$ 
and let 
$\alpha \left( t, h \right) = 0$ otherwise. 
\end{itemize}
This construction guarantees that for all even $t$ we have that 
$h_{t} = 1$ if $t \leq 2 / \gamma$, 
while for even $t > 2 / \gamma$ we have that 
\[
h_{t} = \left\lfloor \gamma t / 2 \right\rfloor
\]
deterministically. 
Thus by~\eqref{eq:Rt_identity} we have, for all even $t$ satisfying $t > 2 / \gamma$, that 
\[
R_{t} = t / 2 - h_{t} = t / 2 - \left\lfloor \gamma t / 2 \right\rfloor 
\geq \left( 1 - \gamma \right) t / 2.
\]
On the other hand, by~\eqref{eq:MLE_correct_prob} we have, for all even $t$ satisfying $t > 2 / \gamma$, that 
\[
\p \left( \wh{v}_{\ML} = v^{*} \right) 
= \frac{1}{d \left( d - 1 \right)^{\left\lfloor \gamma t / 2 \right\rfloor - 1}} 
\leq \frac{1}{d \left( d - 1 \right)^{\gamma t / 2 - 2}}.
\]
From~\eqref{eq:Nt} it follows that 
$\left( d - 1 \right)^{-t/2} \leq \left( d / \left( d - 2 \right) \right) / N_{t}$ and so 
\[
\p \left( \wh{v}_{\ML} = v^{*} \right) 
\leq \frac{\left( d-1 \right)^{2}}{d} \left( \frac{d}{d-2} \right)^{\gamma} \frac{1}{N_{t}^{\gamma}} 
\leq \frac{ \left( d-1 \right)^{2}}{d-2} \cdot \frac{1}{N_{t}^{\gamma}} 
\leq \frac{2 \left( d - 1 \right)}{N_{t}^{\gamma}},
\]
where in the second inequality we used that $\gamma \leq 1$ 
and in the third inequality we used that $d \geq 3$. 
\end{proof}

\section{Discussion} \label{sec:discussion} 

The main message of this work is that while adaptive diffusion protocols can hide the source from a snapshot adversary, they are ineffective when the adversary has access to multiple independent snapshots. 
The main question raised by our work is whether there exist other diffusion protocols that can obfuscate the source against such an adversary.

We make several simplifying assumptions in this work, which are important to discuss and study further. 
First, we assume throughout that the underlying graph is the infinite $d$-regular tree $\mathbb{T}_{d}$ (with $d\geq 3$), which is not a realistic model of real-world (social) networks. It is therefore important to study the questions we consider here on other underlying graphs, for instance, on more realistic models as well as on real-world social networks. 
We conjecture that our qualitative conclusions will carry over to more realistic settings, which motivates studying such a simplified setting. 

We also assume that the adversary observes multiple \emph{independent} snapshots. 
Previous work has considered multiple \emph{sequential} snapshots (in time): 
Wang~et~al.~\cite{wang2014} show that additional sequential snapshots cannot improve detectability under the SI model, 
while Fanti~et~al.~\cite{fanti2017hide} show that they can improve the detection probability at most logarithmically for adaptive diffusions. 
On the other hand, Cai~et~al.~\cite{cai2018information} show that multiple sequential snapshots can help detection when the spreading rates are heterogeneous, both theoretically and on Twitter data. 
As mentioned before, Wang~et~al.~\cite{wang2014} show that multiple independent snapshots help significantly with detection under the SI model, and our results extend this to the family of adaptive diffusions. 
An interesting question is what happens in between, when the adversary observes multiple correlated snapshots (that are not necessarily sequential observations of the same diffusion). 
In particular, can spreading protocols take advantage of correlation in order to obfuscate the source against an adversary who observes multiple snapshots?

This question is related to local spreading as follows. 
An adversary who observes multiple snapshots can always use the following simple source estimator: pick a node uniformly at random among those which are infected in each snapshot. The probability of success of this estimator is the inverse of the size of the set of nodes which are infected in each snapshot. 
To minimize this, a spreading protocol should aim to maximize the size of this set. 
This can be done by having highly correlated snapshots, or by having a large amount of local spreading (which we have discussed in Sections~\ref{sec:local_spreading} and~\ref{sec:local_spreading_proofs}). 
In any case, we conjecture that if there is a reasonable amount of independence among the observed snapshots, then the results will be qualitatively similar to those which we have obtained.

There are also many natural variations on what information the adversary has access to. 
For instance, Fanti~et~al.~\cite{fanti2016metadata,fanti2017hide} consider a spy-based model, where a fraction of nodes are corrupted and continuously monitor metadata such as message timestamps; they also consider a mixed model using both spies and a snapshot. 
Other information models include 
having a snapshot and additional relative information about the infection times of
a fraction of node pairs~\cite{kumar2017temporally}, 
having partial infection timestamps~\cite{tang2018estimating}, 
and having a noisy time series of observations~\cite{sridhar2019sequential,sridhar2020bayes}. 
Understanding how our results change under these different information models of adversaries is a natural question for future work.

Further avenues to explore related to our work include game-theoretic formulations~\cite{luo2016infection}, optimal sensor/spy placement~\cite{spinelli2017general}, confidence sets for the source~\cite{kl16confidence}, and multiple rumor sources~\cite{spencer2015impossibility}. 

In conclusion, most results in this space---including ours in this work---are positive in terms of rumor source detection, and thus highlight major difficulties with guaranteeing anonymity for the source of a message in a social network. 
As surveillance techniques grow ever more prominent in society, 
this emphasizes the need for further research, with the hope of ultimately providing robust anonymity guarantees.


\bibliographystyle{abbrv}
\bibliography{bib}


\appendix 

\section{Proof of Theorem~\ref{thm:two_obs}(\ref{thm:two_obs_obfu}) when one or both of $t_{1}$ and $t_{2}$ are odd}\label{sec:two_obs_obfu_odd} 

Here we prove Theorem~\ref{thm:two_obs}(\ref{thm:two_obs_obfu}) when one or both of $t_{1}$ and $t_{2}$ are odd. 
This is similar to the proof presented in Section~\ref{sec:two_obs_obfu_even} when both $t_{1}$ and $t_{2}$ are even, but there are more cases to consider (especially when both $t_{1}$ and $t_{2}$ are odd). 
To keep things clear, 
we also separate below the proof when one of $t_{1}$ and $t_{2}$ is even and the other is odd, 
from the proof when both $t_{1}$ and $t_{2}$ are odd.

\subsection{When one of $t_{1}$ and $t_{2}$ is even and the other is odd} 

\begin{proof}[Proof of Theorem~\ref{thm:two_obs}(\ref{thm:two_obs_obfu}) when one of $t_{1}$ and $t_{2}$ is even and the other is odd.]
W.l.o.g.\ 
suppose that $t_{1}$ is even and $t_{2}$ is odd. We may (and will) assume that $t_{1} \geq 4$ and that $t_{2} \geq 5$, since if $\min \left\{ t_{1}, t_{2} \right\} \leq 3$, then the right hand side of~\eqref{eq:two_obs_obfu} is greater than $1$ and thus the statement is vacuously true. 

Recall that we are analyzing the MLE for the \emph{uniform protocol} given by the probabilities in~\eqref{eq:alpha_uniform} 
and for which we have that~\eqref{eq:p_uniform} holds. 
We start with a few observations. 
First, $G_{t_{1}}^{1}$ is a ball of radius $t_{1} / 2$ around $\vs^{1} \equiv \vs_{t_{1}}^{1}$ and hence we can determine $\vs^{1}$. 
We may thus define $X_{1}(v) := \delta \left( v,\vs^{1} \right)$. 
Next, the observation $G_{t_{2}}^{2}$ is either a ball or it is not. 
If $G_{t_{2}}^{2}$ is a ball, then it is a ball of radius $(t_{2}-1)/2$ around $\vs_{t_{2} - 1}^{2} = \vs_{t_{2}}^{2}$ and hence we can determine $\vs_{t_{2} - 1}^{2} = \vs_{t_{2}}^{2}$. 
If $G_{t_{2}}^{2}$ is not a ball, then its central edge is $\left\{ \vs_{t_{2} - 1}^{2}, \vs_{t_{2}}^{2} \right\}$ and hence we can determine the set $\left\{ \vs_{t_{2} - 1}^{2}, \vs_{t_{2}}^{2} \right\}$. 
In any case, we may thus define 
$X_{2}(v) := \min \left\{ \delta \left( v, \vs_{t_{2} - 1}^{2} \right), \delta \left( v, \vs_{t_{2}}^{2} \right) \right\}$ 
and note that the function $X_{2} : V \mapsto \R$ is determined by the observation $G_{t_{2}}^{2}$.

With this notation, 
using the expressions~\eqref{eq:likelihood_even} and~\eqref{eq:likelihood_odd}, 
the independence of the two observations, 
and substituting the expressions in~\eqref{eq:alpha_uniform} and~\eqref{eq:p_uniform}, 
we obtain that the likelihood function is the following: 
\begin{equation*}\label{eq:likelihood_even_and_odd}
 L(v) = 
 \begin{cases}
  \left( \frac{d-1}{d} \right)^{2} \cdot \frac{2}{t_{1}} \cdot \frac{2}{t_{2} - 1} \cdot \left( d - 1 \right)^{- \left( X_{1}(v) + X_{2}(v) \right)} \cdot \frac{t_{2}+1 - 2 X_{2}(v)}{t_{2} + 1} \cdot \mathbf{1}_{\left\{ v \in V_{t_{1}}^{1} \cap V_{t_{2}}^{2} \setminus \left\{ \vs_{t_{1}}^{1}, \vs_{t_{2}}^{2} \right\} \right\}} 
  &\text{ if } G_{t_{2}}^{2} \text{ is a ball}, \\
  \left( \frac{d-1}{d} \right)^{2} \cdot \frac{2}{t_{1}} \cdot \frac{2}{t_{2} - 1} \cdot \left( d - 1 \right)^{- \left( X_{1}(v) + X_{2}(v) \right)} \cdot \frac{2 X_{2}(v)}{t_{2}+1} \cdot \mathbf{1}_{\left\{ v \in V_{t_{1}}^{1} \cap V_{t_{2}}^{2} \setminus \left\{ \vs_{t_{1}}^{1}, \vs_{t_{2}}^{2}, \vs_{t_{2}-1}^{2} \right\} \right\}}
  &\text{ otherwise}.
 \end{cases}
\end{equation*}
The first three factors in the expressions above do not depend on $v$ and hence do not matter for the MLE. 
By taking logarithms and pulling out a minus sign, we obtain that the MLE satisfies the following: 
\begin{equation}\label{eq:MLE_even_and_odd}
 \wh{v}_{\ML} \in 
 \begin{cases}
  \argmin\limits_{v \in V_{t_{1}}^{1} \cap V_{t_{2}}^{2} \setminus \left\{ \vs_{t_{1}}^{1}, \vs_{t_{2}}^{2} \right\}} 
  \left[ \left( X_{1}(v) + X_{2}(v) \right) \log(d-1) - \log \left( t_{2} + 1 - 2 X_{2} (v) \right) \right] 
  &\text{ if } G_{t_{2}}^{2} \text{ is a ball}, \\
  \argmin\limits_{v \in V_{t_{1}}^{1} \cap V_{t_{2}}^{2} \setminus \left\{ \vs_{t_{1}}^{1}, \vs_{t_{2}}^{2}, \vs_{t_{2}-1}^{2} \right\}} 
  \left[ \left( X_{1}(v) + X_{2}(v) \right) \log(d-1) - \log X_{2}(v) \right] 
  &\text{ otherwise}.
 \end{cases}
\end{equation}

To understand the MLE in~\eqref{eq:MLE_even_and_odd} better, we distinguish six cases. 
These are based on whether $X_{2}(\vs^{1})$ is $0$, $1$, or at least $2$, 
and whether or not $G_{t_{2}}^{2}$ is a ball. 
Compared to the case when both $t_{1}$ and $t_{2}$ are even (see Section~\ref{sec:two_obs_obfu_even}) we have twice as many cases to consider because $G_{t_{2}}^{2}$ can be a ball or not. 
In all of the six cases there is a simple description of the MLE:
\begin{enumerate}[(1)]
 \item\label{case:even_odd_ball_0} 
 \emph{$G_{t_{2}}^{2}$ is a ball and $X_{2}(\vs^{1}) = 0$.} 
 In this case $\vs^{1} = \vs_{t_{2}}^{2} = \vs_{t_{2}-1}^{2}$ 
 and both terms in~\eqref{eq:MLE_even_and_odd} 
 ($\left( X_{1}(v) + X_{2}(v) \right) \log(d-1)$ and $- \log \left( t_{2} + 1 - 2 X_{2} (v) \right)$) 
 are minimized by the neighbors of~$\vs^{1}$. 
 Therefore $\wh{v}_{\ML}$ chooses a neighbor of $\vs^{1}$ uniformly at random. 
 
 \item\label{case:even_odd_notball_0} 
 \emph{$G_{t_{2}}^{2}$ is not a ball and $X_{2}(\vs^{1}) = 0$.} 
 In this case $\vs_{t_{2} - 1}^{2}$ and $\vs_{t_{2}}^{2}$ are neighbors 
 and $\vs^{1} \in \left\{ \vs_{t_{2} - 1}^{2}, \vs_{t_{2}}^{2} \right\}$. 
 The function $\left( X_{1}(v) + X_{2}(v) \right) \log(d-1)$ is minimized (among vertices not in $\left\{ \vs_{t_{1}}^{1}, \vs_{t_{2}}^{2}, \vs_{t_{2}-1}^{2} \right\}$)
 by the neighbors of $\vs^{1}$ that are not in $\left\{ \vs_{t_{2}}^{2}, \vs_{t_{2}-1}^{2} \right\}$. 
 The expression in~\eqref{eq:MLE_even_and_odd} also contains another term, $- \log X_{2}(v)$, 
 which is not minimized among the neighbors of $\vs^{1}$. 
 However, this is a lower order term compared to the first term 
 and it can be seen that the full expression in~\eqref{eq:MLE_even_and_odd} is minimized 
 by the neighbors of $\vs^{1}$ that are not in $\left\{ \vs_{t_{2}}^{2}, \vs_{t_{2}-1}^{2} \right\}$. 
 Therefore $\wh{v}_{\ML}$ chooses a neighbor of $\vs^{1}$ that is not in $\left\{ \vs_{t_{2}}^{2}, \vs_{t_{2}-1}^{2} \right\}$ 
 uniformly at random.

 \item\label{case:even_odd_ball_1} 
 \emph{$G_{t_{2}}^{2}$ is a ball and $X_{2}(\vs^{1}) = 1$.} 
 In this case $\vs_{t_{2} - 1}^{2} = \vs_{t_{2}}^{2} \equiv \vs^{2}$, 
 and $\vs^{1}$ and $\vs^{2}$ are neighbors. 
 The first term in~\eqref{eq:MLE_even_and_odd} 
 ($\left( X_{1}(v) + X_{2}(v) \right) \log(d-1)$) 
 is minimized by the neighbors of the set $\left\{ \vs^{1}, \vs^{2} \right\}$. 
 The second term in~\eqref{eq:MLE_even_and_odd} 
 ($- \log \left( t_{2} + 1 - 2 X_{2} (v) \right)$) 
 is minimized by the neighbors of $\vs^{2}$ that are not $\vs^{1}$. 
 Thus the whole expression is minimized by the neighbors of $\vs^{2}$ that are not $\vs^{1}$. 
 Therefore $\wh{v}_{\ML}$ chooses one of the $d-1$ neighbors of $\vs^{2}$ that is not $\vs^{1}$, 
 uniformly at random.\footnote{Here we use that $t_{1} \geq 4$, to ensure that all neighbors of $\vs^{2}$ are in $V_{t_{1}}^{1}$.} 
 
 \item\label{case:even_odd_notball_1} 
 \emph{$G_{t_{2}}^{2}$ is not a ball and $X_{2}(\vs^{1}) = 1$.} 
 This is similar to Case~\eqref{case:even_odd_notball_0} above and we omit the details for brevity. 
 The MLE is the same: 
 $\wh{v}_{\ML}$ chooses one of the $d-1$ neighbors of $\vs^{1}$ that is not in $\left\{ \vs_{t_{2}}^{2}, \vs_{t_{2}-1}^{2} \right\}$,  
 uniformly at random. 
 
 \item\label{case:even_odd_ball_2} 
 \emph{$G_{t_{2}}^{2}$ is a ball and $X_{2}(\vs^{1}) \geq 2$.} 
 In this case $\vs_{t_{2} - 1}^{2} = \vs_{t_{2}}^{2} \equiv \vs^{2}$. 
 The expression in~\eqref{eq:MLE_even_and_odd} contains two terms. 
 The first term ($\left( X_{1}(v) + X_{2}(v) \right) \log(d-1)$) 
 is minimized on the shortest path between $\vs^{1}$ and $\vs^{2}$. 
 The second term ($- \log \left( t_{2} + 1 - 2 X_{2} (v) \right)$) 
 is minimized when $X_{2}(v)$ is minimized. 
 We also have the constraint that 
 $v \in V_{t_{1}}^{1} \cap V_{t_{2}}^{2} \setminus \left\{ \vs^{1}, \vs^{2} \right\}$. 
 Let $P_{12}$ denote the set of vertices that are on the shortest path between $\vs^{1}$ and $\vs^{2}$, excluding $\vs^{1}$ and $\vs^{2}$. 
 Define $v'$ to be the vertex in $P_{12} \cap V_{t_{1}}^{1} \cap V_{t_{2}}^{2}$ that is closest to $\vs^{2}$ (note that $P_{12} \cap V_{t_{1}}^{1} \cap V_{t_{2}}^{2}$ is nonempty, due to the assumption that $X_{2}(\vs^{1}) \geq 2$). 
 Given the constraints, $v'$ is the unique vertex that minimizes both terms in the expression in~\eqref{eq:MLE_even_and_odd}. 
 Therefore $\wh{v}_{\ML} = v'$.

 \item\label{case:even_odd_notball_2}
 \emph{$G_{t_{2}}^{2}$ is not a ball and $X_{2}(\vs^{1}) \geq 2$.} 
 This is similar to Case~\eqref{case:even_odd_ball_2} above, so we omit the details for brevity and just state the conclusion. 
 Let $v''$ be the vertex in $P_{12} \cap V_{t_{1}}^{1} \cap V_{t_{2}}^{2}$ that is closest to $\vs^{1}$ 
 (note again that $P_{12} \cap V_{t_{1}}^{1} \cap V_{t_{2}}^{2}$ is nonempty, due to the assumption that $X_{2}(\vs^{1}) \geq 2$). 
 Then $\wh{v}_{\ML} = v''$. 
\end{enumerate}

Now that we understand the MLE, we can compute the probability that it is correct. 
We may again condition on whether or not $A_{12}$ holds (see~\eqref{eq:MLE_total_prob}) to obtain that 
\begin{equation}\label{eq:MLE_total_prob_even_odd}
 \p \left( \wh{v}_{\ML} = v^{*} \right) 
 = \p \left( \wh{v}_{\ML} = v^{*} \, \middle| \, A_{12} \right) \cdot \frac{d-1}{d}
 + \p \left( \wh{v}_{\ML} = v^{*} \, \middle| \, A_{12}^{C} \right) \cdot \frac{1}{d}.
\end{equation}
Note that if $A_{12}$ holds then we must be in Case~\eqref{case:even_odd_ball_2} or in Case~\eqref{case:even_odd_notball_2}. 
Consequently, Cases~\eqref{case:even_odd_ball_0},~\eqref{case:even_odd_notball_0},~\eqref{case:even_odd_ball_1}, and~\eqref{case:even_odd_notball_1} imply that $A_{12}$ does not hold.

Let us start by computing 
$\p \left( \wh{v}_{\ML} = v^{*} \, \middle| \, A_{12}^{C} \right)$. 
Recall from Cases~\eqref{case:even_odd_ball_2} and~\eqref{case:even_odd_notball_2} above 
that if $X_{2}(\vs^{1}) \geq 2$ then $\wh{v}_{\ML} \in P_{12}$, 
so if $A_{12}^{C}$ also holds then $\wh{v}_{\ML} \neq v^{*}$. 
Thus only Cases~\eqref{case:even_odd_ball_0},~\eqref{case:even_odd_notball_0},~\eqref{case:even_odd_ball_1}, and~\eqref{case:even_odd_notball_1} contribute to the probability 
$\p \left( \wh{v}_{\ML} = v^{*} \, \middle| \, A_{12}^{C} \right)$. 
In the following computations we use the expressions~\eqref{eq:alpha_uniform} and~\eqref{eq:p_uniform}. 

First, corresponding to Case~\eqref{case:even_odd_ball_0} above: 
if $\delta \left( v^{*}, \vs_{t_{1}}^{1} \right) = \delta \left( v^{*}, \vs_{t_{2} - 1}^{2} \right) = 1$ 
and $\vs_{t_{2} - 1}^{2} = \vs_{t_{2}}^{2}$, 
then the MLE is correct with probability $1/d$. 
This gives a contribution of 
$\tfrac{2}{t_{1}} \cdot \tfrac{2}{t_{2} - 1} \cdot \tfrac{t_{2} - 1}{t_{2} + 1} \cdot \tfrac{1}{d} 
= \tfrac{4}{t_{1} \left( t_{2} + 1 \right)} \cdot \tfrac{1}{d}$. 

Second, corresponding to Case~\eqref{case:even_odd_notball_0} above: 
if $\delta \left( v^{*}, \vs_{t_{1}}^{1} \right) = \delta \left( v^{*}, \vs_{t_{2} - 1}^{2} \right) = 1$ 
and $\delta \left( v^{*}, \vs_{t_{2}}^{2} \right) = 2$, 
then the MLE is correct with probability $1/(d-1)$. 
This gives a contribution of 
$\tfrac{2}{t_{1}} \cdot \tfrac{2}{t_{2}-1} \cdot \tfrac{2}{t_{2}+1} \cdot \tfrac{1}{d-1} 
= \tfrac{4}{t_{1} \left( t_{2} + 1 \right)} \cdot \tfrac{2}{t_{2} - 1} \cdot \tfrac{1}{d-1}$. 

Next, corresponding to Case~\eqref{case:even_odd_ball_1} above: 
if $\delta \left( v^{*}, \vs_{t_{2} - 1}^{2} \right) = 1$, 
$\delta \left( \vs_{t_{2} - 1}^{2}, \vs_{t_{1}}^{1} \right) = 1$, 
and $\vs_{t_{2} - 1}^{2} = \vs_{t_{2}}^{2}$, 
then the MLE is correct with probability $1/(d-1)$. 
This gives a contribution of 
$\tfrac{2}{t_{1}} \cdot \tfrac{2}{t_{2} - 1} \cdot \tfrac{t_{2} - 1}{t_{2} + 1} \cdot \tfrac{1}{d-1} 
= \tfrac{4}{t_{1} \left( t_{2} + 1 \right)} \cdot \tfrac{1}{d-1}$. 

Finally, corresponding to Case~\eqref{case:even_odd_notball_1} above: 
if $\delta \left( v^{*}, \vs_{t_{1}}^{1} \right) = 1$, 
$\delta \left( v^{*}, \vs_{t_{2} - 1}^{2} \right) = 2$, 
$\delta \left( v^{*}, \vs_{t_{2}}^{2} \right) = 3$, 
then the MLE is correct with probability $1/(d-1)$. 
This gives a contribution of 
$\tfrac{2}{t_{1}} \cdot \tfrac{2}{t_{2} - 1} \cdot \tfrac{4}{t_{2} + 1} \cdot \tfrac{1}{d-1} 
= \tfrac{4}{t_{1} \left( t_{2} + 1 \right)} \cdot \tfrac{4}{t_{2} - 1} \cdot \tfrac{1}{d-1}$. 

Putting together these four contributions, we obtain that 
\begin{equation}\label{eq:MLE_correct_even_odd_A12C}
 \p \left( \wh{v}_{\ML} = v^{*} \, \middle| \, A_{12}^{C} \right) 
 = \frac{4}{t_{1} \left( t_{2} + 1 \right)} \left( \frac{1}{d} + \frac{1}{d-1} + \frac{6}{t_{2} - 1} \cdot \frac{1}{d-1} \right). 
\end{equation}

We now turn to computing 
$\p \left( \wh{v}_{\ML} = v^{*} \, \middle| \, A_{12} \right)$. 
As mentioned before, if $A_{12}$ holds, then we must be in Case~\eqref{case:even_odd_ball_2} or in Case~\eqref{case:even_odd_notball_2} above.  
Accordingly, we have contributions to the probability 
$\p \left( \wh{v}_{\ML} = v^{*} \, \middle| \, A_{12} \right)$ 
from these two cases and we can break this probability into two terms: 
\begin{equation}\label{eq:cases_5_and_6}
 \p \left( \wh{v}_{\ML} = v^{*} \, \middle| \, A_{12} \right) 
 = \p \left( \wh{v}_{\ML} = v^{*}, \vs_{t_{2} - 1}^{2} = \vs_{t_{2}}^{2} \, \middle| \, A_{12} \right) 
 + \p \left( \wh{v}_{\ML} = v^{*}, \vs_{t_{2} - 1}^{2} \neq \vs_{t_{2}}^{2} \, \middle| \, A_{12} \right) 
\end{equation}
Here the first term corresponds to Case~\eqref{case:even_odd_ball_2} and the second corresponds to Case~\eqref{case:even_odd_notball_2}. 

Recall that in Case~\eqref{case:even_odd_ball_2} we have that $\wh{v}_{\ML} = v'$, where $v'$ is the vertex in $P_{12} \cap V_{t_{1}}^{1} \cap V_{t_{2}}^{2}$ that is closest to $\vs^{2}$. 
Given that $A_{12}$ holds, we have that $v^{*} = v'$ in exactly two situations: 
if $\delta \left( v^{*}, \vs^{2} \right) = 1$ 
or if $\delta \left( v^{*}, \vs^{1} \right) = t_{1} / 2$. 
In addition, to be in Case~\eqref{case:even_odd_ball_2} we have to have $\vs_{t_{2} - 1}^{2} = \vs_{t_{2}}^{2}$ (the other condition, $X_{2}(\vs^{1}) \geq 2$, is guaranteed given $A_{12}$). By separating into these cases we can write: 
\begin{multline}
 \p \left( \wh{v}_{\ML} = v^{*}, \vs_{t_{2} - 1}^{2} = \vs_{t_{2}}^{2} \, \middle| \, A_{12} \right) \\
\begin{aligned}
 &= \p \left( \left( \left\{ \delta \left( v^{*}, \vs_{t_{2}}^{2} \right) = 1 \right\} \cup \left\{ \delta \left( v^{*}, \vs_{t_{1}}^{1} \right) = t_{1} / 2 \right\} \right) \cap \left\{ \vs_{t_{2} - 1}^{2} = \vs_{t_{2}}^{2} \right\} \, \middle| \, A_{12} \right) \\
 &= \p \left( \delta \left( v^{*}, \vs_{t_{2}}^{2} \right) = 1 \, \middle| \, A_{12} \right) 
 + \p \left( \delta \left( v^{*}, \vs_{t_{1}}^{1} \right) = t_{1} / 2, \vs_{t_{2} - 1}^{2} = \vs_{t_{2}}^{2} \, \middle| \, A_{12} \right) \\
 &\quad - \p \left( \delta \left( v^{*}, \vs_{t_{2}}^{2} \right) = 1, \delta \left( v^{*}, \vs_{t_{1}}^{1} \right) = t_{1} / 2 \, \middle| \, A_{12} \right) \label{eq:three_probs}
\end{aligned}
\end{multline}
We now compute all three of these probabilities. Before we do so, we first compute the probability that $\vs_{t_{2} - 1}^{2} = \vs_{t_{2}}^{2}$. We do this by conditioning on the value of $\delta \left( v^{*}, \vs_{t_{2} - 1}^{2} \right)$. Using~\eqref{eq:alpha_uniform} and~\eqref{eq:p_uniform} we have that 
\begin{align*}
\p \left( \vs_{t_{2} - 1}^{2} = \vs_{t_{2}}^{2} \right) 
&= \sum_{h = 1}^{(t_{2} - 1)/2} \p \left( \vs_{t_{2} - 1}^{2} = \vs_{t_{2}}^{2} \, \middle| \, \delta \left( v^{*}, \vs_{t_{2} - 1}^{2} \right) = h \right) \cdot \frac{2}{t_{2} - 1} \\
&= \sum_{h = 1}^{(t_{2} - 1)/2} \frac{t_{2} + 1 - 2 h}{t_{2} + 1} \cdot \frac{2}{t_{2} - 1} 
= 1 - \frac{4}{(t_{2}-1)(t_{2}+1)} \sum_{h = 1}^{(t_{2} - 1)/2} h 
= 1 - \frac{1}{2} = \frac{1}{2}.
\end{align*}
Turning back to the three probabilities in~\eqref{eq:three_probs}, 
notice that in all three cases the events are independent of $A_{12}$. 
So first we have that 
\[
\p \left( \delta \left( v^{*}, \vs_{t_{2}}^{2} \right) = 1 \, \middle| \, A_{12} \right) 
= \p \left( \delta \left( v^{*}, \vs_{t_{2}}^{2} \right) = 1 \right) 
= \frac{2}{t_{2}+1}.
\]
Next, by independence we have that 
\begin{align*}
\p \left( \delta \left( v^{*}, \vs_{t_{1}}^{1} \right) = t_{1} / 2, \vs_{t_{2} - 1}^{2} = \vs_{t_{2}}^{2} \, \middle| \, A_{12} \right) 
&= \p \left( \delta \left( v^{*}, \vs_{t_{1}}^{1} \right) = t_{1} / 2, \vs_{t_{2} - 1}^{2} = \vs_{t_{2}}^{2} \right) \\ 
&= \p \left( \delta \left( v^{*}, \vs_{t_{1}}^{1} \right) = t_{1} / 2 \right) \p \left( \vs_{t_{2} - 1}^{2} = \vs_{t_{2}}^{2} \right) 
= \frac{2}{t_{1}} \cdot \frac{1}{2} = \frac{1}{t_{1}}.
\end{align*}
Finally, by using independence again, we have that 
\begin{multline*}
\p \left( \delta \left( v^{*}, \vs_{t_{2}}^{2} \right) = 1, \delta \left( v^{*}, \vs_{t_{1}}^{1} \right) = t_{1} / 2 \, \middle| \, A_{12} \right) 
= \p \left( \delta \left( v^{*}, \vs_{t_{2}}^{2} \right) = 1, \delta \left( v^{*}, \vs_{t_{1}}^{1} \right) = t_{1} / 2 \right) \\
= \p \left( \delta \left( v^{*}, \vs_{t_{2}}^{2} \right) = 1 \right) \p \left( \delta \left( v^{*}, \vs_{t_{1}}^{1} \right) = t_{1} / 2 \right) 
= \frac{2}{t_{2} + 1} \cdot \frac{2}{t_{1}}.
\end{multline*}
Plugging the previous three displays into~\eqref{eq:three_probs}, we have determined the first term in~\eqref{eq:cases_5_and_6}:
\begin{equation}\label{eq:cases_5_and_6_first_term}
\p \left( \wh{v}_{\ML} = v^{*}, \vs_{t_{2} - 1}^{2} = \vs_{t_{2}}^{2} \, \middle| \, A_{12} \right) 
= \frac{1}{t_{1}} + \frac{2}{t_{2} + 1} - \frac{4}{t_{1} \left( t_{2} + 1 \right)}.
\end{equation}

The analysis of the second term in~\eqref{eq:cases_5_and_6} is analogous to what we have just done for the first term, so we omit the details. In fact, it turns out that this second term is equal to the first term: 
\begin{equation}\label{eq:cases_5_and_6_second_term}
\p \left( \wh{v}_{\ML} = v^{*}, \vs_{t_{2} - 1}^{2} \neq \vs_{t_{2}}^{2} \, \middle| \, A_{12} \right)
= \frac{1}{t_{1}} + \frac{2}{t_{2} + 1} - \frac{4}{t_{1} \left( t_{2} + 1 \right)}.
\end{equation}
Putting together~\eqref{eq:cases_5_and_6},~\eqref{eq:cases_5_and_6_first_term}, and~\eqref{eq:cases_5_and_6_second_term}, we obtain that 
\begin{equation}\label{eq:MLE_correct_even_odd_A12}
 \p \left( \wh{v}_{\ML} = v^{*} \, \middle| \, A_{12} \right) 
 = \frac{2}{t_{1}} + \frac{4}{t_{2} + 1} - \frac{8}{t_{1} \left( t_{2} + 1 \right)}
\end{equation}

Thus putting together~\eqref{eq:MLE_total_prob_even_odd},~\eqref{eq:MLE_correct_even_odd_A12C}, and~\eqref{eq:MLE_correct_even_odd_A12}, we obtain that 
\[
 \p \left( \wh{v}_{\ML} = v^{*} \right) 
 = \frac{d-1}{d} \left( \frac{2}{t_{1}} + \frac{4}{t_{2} + 1} - \frac{8}{t_{1} \left( t_{2} + 1 \right)} \right) 
 + \frac{1}{d} \cdot \frac{4}{t_{1} \left( t_{2} + 1 \right)} \left( \frac{1}{d} + \frac{1}{d-1} + \frac{6}{t_{2} - 1} \cdot \frac{1}{d-1} \right).
\]
Separating the main terms and the lower order terms, we can write 
\begin{align*}
 \p \left( \wh{v}_{\ML} = v^{*} \right) 
 &= \frac{d-1}{d} \left( \frac{2}{t_{1}} + \frac{4}{t_{2} + 1} \right) \\
 &\quad + \frac{1}{d} \cdot \frac{4}{t_{1} \left( t_{2} + 1 \right)} \left( - 2 \left( d - 1 \right) +  \frac{1}{d} + \frac{1}{d-1} + \frac{6}{t_{2} - 1} \cdot \frac{1}{d-1} \right).
\end{align*}
Since $d \geq 3$ and $t_{2} \geq 5$, the second term in the expression above is negative, so we have that 
\[
 \p \left( \wh{v}_{\ML} = v^{*} \right) 
 \leq \frac{d-1}{d} \left( \frac{2}{t_{1}} + \frac{4}{t_{2} + 1} \right) 
 \leq \frac{d-1}{d} \frac{6}{\min \left\{ t_{1}, t_{2} \right\}}. \qedhere
\]
\end{proof}

\subsection{When both $t_{1}$ and $t_{2}$ are odd} 

\begin{proof}[Proof of Theorem~\ref{thm:two_obs}(\ref{thm:two_obs_obfu}) when both $t_{1}$ and $t_{2}$ are odd.] 
We may (and will) assume in the following that $t_{1}, t_{2} \geq 5$, since if $\min \left\{ t_{1}, t_{2} \right\} \leq 3$, then the right hand side of~\eqref{eq:two_obs_obfu} is greater than $1$ and thus the statement is vacuously true. 

Recall that we are analyzing the MLE for the \emph{uniform protocol} given by the probabilities in~\eqref{eq:alpha_uniform} 
and for which we have that~\eqref{eq:p_uniform} holds. 
We start with a few observations. 
For $i \in \left\{ 1, 2 \right\}$, 
the observation $G_{t_{i}}^{i}$ is either a ball or it is not. 
If $G_{t_{i}}^{i}$ is a ball, then it is a ball of radius $(t_{i}-1)/2$ around $\vs_{t_{i} - 1}^{i} = \vs_{t_{i}}^{i}$, and hence we can determine $\vs_{t_{i} - 1}^{i} = \vs_{t_{i}}^{i}$. 
If $G_{t_{i}}^{i}$ is not a ball, 
then its central edge is $\left\{ \vs_{t_{i} - 1}^{i}, \vs_{t_{i}}^{i} \right\}$, 
and hence we can determine the set $\left\{ \vs_{t_{i} - 1}^{i}, \vs_{t_{i}}^{i} \right\}$. 
In any case, we may thus define 
$X_{i}(v) := \min \left\{ \delta \left( v, \vs_{t_{i} - 1}^{i} \right), \delta \left( v, \vs_{t_{i}}^{i} \right) \right\}$ 
and note that the function $X_{i} : V \mapsto \R$ is determined by the observation $G_{t_{i}}^{i}$.

With this notation, 
using the expressions in~\eqref{eq:likelihood_odd}, 
the independence of the two observations, 
and substituting the expressions in~\eqref{eq:alpha_uniform} and~\eqref{eq:p_uniform}, 
we can write down the likelihood function. 
There are four cases, depending on whether or not $G_{t_{1}}^{1}$ and $G_{t_{2}}^{2}$ are balls or not. 
Instead of detailing all four cases of the likelihood function, we skip straight to writing down an expression for the MLE in the four cases; this is analogous to~\eqref{eq:MLE_even_and_odd}. 
If $G_{t_{1}}^{1}$ and $G_{t_{2}}^{2}$ are both balls, then 
\begin{equation}\label{eq:MLE_odd_odd_balls}
 \wh{v}_{\ML} \in 
 \argmin\limits_{v \in V_{t_{1}}^{1} \cap V_{t_{2}}^{2} \setminus \left\{ \vs_{t_{1}}^{1}, \vs_{t_{2}}^{2} \right\}} 
  \left[ \left( X_{1}(v) + X_{2}(v) \right) \log(d-1) - \log \left( t_{1} + 1 - 2 X_{1} (v) \right) - \log \left( t_{2} + 1 - 2 X_{2} (v) \right) \right]. 
\end{equation}
If $G_{t_{1}}^{1}$ is a ball and $G_{t_{2}}^{2}$ is not a ball, then 
\begin{equation}\label{eq:MLE_odd_odd_ball_notball}
 \wh{v}_{\ML} \in 
 \argmin\limits_{v \in V_{t_{1}}^{1} \cap V_{t_{2}}^{2} \setminus \left\{ \vs_{t_{1}}^{1}, \vs_{t_{2}}^{2}, \vs_{t_{2} - 1}^{2} \right\}} 
  \left[ \left( X_{1}(v) + X_{2}(v) \right) \log(d-1) - \log \left( t_{1} + 1 - 2 X_{1} (v) \right) - \log  X_{2} (v) \right]. 
\end{equation}
If $G_{t_{1}}^{1}$ is not a ball and $G_{t_{2}}^{2}$ is a ball, 
then the MLE satisfies the display above with the indices $1$ and $2$ switched. 
Finally, if $G_{t_{1}}^{1}$ is not a ball and $G_{t_{2}}^{2}$ is also not a ball, then 
\begin{equation}\label{eq:MLE_odd_odd_notballs}
 \wh{v}_{\ML} \in 
 \argmin\limits_{v \in V_{t_{1}}^{1} \cap V_{t_{2}}^{2} \setminus \left\{ \vs_{t_{1}}^{1}, \vs_{t_{1} - 1}^{1}, \vs_{t_{2}}^{2}, \vs_{t_{2} - 1}^{2} \right\}} 
  \left[ \left( X_{1}(v) + X_{2}(v) \right) \log(d-1) - \log X_{1} (v) - \log  X_{2} (v) \right]. 
\end{equation}

To understand the MLE better, we distinguish several cases. 
These are based on whether or not $G_{t_{1}}^{1}$ and $G_{t_{2}}^{2}$ are balls, 
as well as the distance of the sets 
$\left\{ \vs_{t_{1} - 1}^{1}, \vs_{t_{1}}^{1} \right\}$ 
and 
$\left\{ \vs_{t_{2} - 1}^{2}, \vs_{t_{2}}^{2} \right\}$ 
(and if their distance is zero, what does their intersection look like). 
Compared to the case where one of $t_{1}$ and $t_{2}$ is even and the other is odd, 
we have essentially twice as many cases to consider, because both $G_{t_{1}}^{1}$ and $G_{t_{2}}^{2}$ can be balls or not. 
In all cases there is a simple description of the MLE, which we detail next. 

We first consider the case when $G_{t_{1}}^{1}$ and $G_{t_{2}}^{2}$ are both balls. 
In this case, to abbreviate notation, we let 
$\vs^{1} \equiv \vs^{1}_{t_{1} - 1} = \vs^{1}_{t_{1}}$ 
and 
$\vs^{2} \equiv \vs^{2}_{t_{2} - 1} = \vs^{2}_{t_{2}}$. 
We distinguish three subcases based on whether $X_{2} \left( \vs^{1} \right)$ is $0$, $1$, or at least $2$. 
\begin{enumerate}[(1)]
\item \textit{$G_{t_{1}}^{1}$ and $G_{t_{2}}^{2}$ are both balls, and $X_{2} \left( \vs^{1} \right) = 0$.} \label{case:odd_odd_balls_1}
In this case $\vs^{1} = \vs^{2}$ and all three terms in~\eqref{eq:MLE_odd_odd_balls} are minimized by the neighbors of $\vs^{1} = \vs^{2}$. 
Therefore $\wh{v}_{\ML}$ chooses a neighbor of $\vs^{1} = \vs^{2}$ uniformly at random. 

\item \textit{$G_{t_{1}}^{1}$ and $G_{t_{2}}^{2}$ are both balls, and $X_{2} \left( \vs^{1} \right) = 1$.} \label{case:odd_odd_balls_2}
In this case $\vs^{1}$ and $\vs^{2}$ are neighbors. Let us consider all three terms in~\eqref{eq:MLE_odd_odd_balls}. 
The function $\left( X_{1}(v) + X_{2}(v) \right) \log(d-1)$ is minimized (among vertices not in $\left\{ \vs^{1}, \vs^{2} \right\}$) by the neighbors of the set $\left\{ \vs^{1}, \vs^{2} \right\}$. 
The function $- \log \left( t_{1} + 1 - 2 X_{1} (v) \right)$ is increasing in $X_{1}(v)$, 
while the function $- \log \left( t_{2} + 1 - 2 X_{2} (v) \right)$ is increasing in $X_{2}(v)$. 
Consequently, the set of minimizers of their sum (among vertices not in $\left\{ \vs^{1}, \vs^{2} \right\}$) is contained within the neighbors of the set $\left\{ \vs^{1}, \vs^{2} \right\}$. 
The precise set of minimizers depends on the relationship between $t_{1}$ and $t_{2}$. 
If $t_{1} = t_{2}$, then $\wh{v}_{\ML}$ chooses a neighbor of the set $\left\{ \vs^{1}, \vs^{2} \right\}$ uniformly at random. 
If $t_{1} > t_{2}$, then $\wh{v}_{\ML}$ chooses one of the $d-1$ neighbors of $\vs^{2}$ that is not $\vs^{1}$, uniformly at random. 
If $t_{1} < t_{2}$, then the indices are switched: $\wh{v}_{\ML}$ chooses one of the $d-1$ neighbors of $\vs^{1}$ that is not $\vs^{2}$, uniformly at random.

\item \textit{$G_{t_{1}}^{1}$ and $G_{t_{2}}^{2}$ are both balls, and $X_{2} \left( \vs^{1} \right) \geq 2$.} \label{case:odd_odd_balls_3} 
We again consider the three terms in~\eqref{eq:MLE_odd_odd_balls}. 
The first term ($\left( X_{1}(v) + X_{2}(v) \right) \log(d-1)$) is minimized on the shortest path between $\vs^{1}$ and~$\vs^{2}$. 
The other two terms are increasing in $X_{1}(v)$ and $X_{2}(v)$, respectively. 
This implies that the minimizer of the whole expression in~\eqref{eq:MLE_odd_odd_balls} lies on the shortest path between $\vs^{1}$ and $\vs^{2}$. 
Let $P_{12}$ denote the set of vertices that are on the shortest path between $\vs^{1}$ and $\vs^{2}$, excluding $\vs^{1}$ and $\vs^{2}$. 
We thus have that the MLE satisfies 
\[
\wh{v}_{\ML} \in 
\argmax_{v \in P_{12} \cap V_{t_{1}}^{1} \cap V_{t_{2}}^{2}} 
\left( t_{1} + 1 - 2 X_{1}(v) \right) \left( t_{2} + 1 - 2 X_{2} \left( v \right) \right).
\]

\end{enumerate}

Next, we consider the case when $G_{t_{1}}^{1}$ is a ball and $G_{t_{2}}^{2}$ is not a ball. In this case we can again write $\vs^{1} \equiv \vs^{1}_{t_{1} - 1} = \vs^{1}_{t_{1}}$ to abbreviate notation. 
We again distinguish three subcases based on whether $X_{2} \left( \vs^{1} \right)$ is $0$, $1$, or at least $2$. 

\begin{enumerate}[(1)]
\setcounter{enumi}{3}
\item \textit{$G_{t_{1}}^{1}$ is a ball, $G_{t_{2}}^{2}$ is not a ball, and $X_{2} \left( \vs^{1} \right) = 0$.} \label{case:odd_odd_ball_notball_4}
In this case $\vs_{t_{2} - 1}^{2}$ and $\vs_{t_{2}}^{2}$ are neighbors 
and $\vs^{1} \in \left\{ \vs_{t_{2} - 1}^{2}, \vs_{t_{2}}^{2} \right\}$. 
Let us consider all three terms in~\eqref{eq:MLE_odd_odd_ball_notball}. 
The first term, 
$\left( X_{1}(v) + X_{2}(v) \right) \log(d-1)$, 
is minimized (among vertices not in $\left\{ \vs_{t_{1}}^{1}, \vs_{t_{2} - 1}^{2}, \vs_{t_{2}}^{2} \right\}$) by the neighbors of $\vs^{1}$ that are not in $\left\{ \vs_{t_{2} - 1}^{2}, \vs_{t_{2}}^{2} \right\}$. 
The second term, 
$- \log \left( t_{1} + 1 - 2 X_{1} (v) \right)$, 
is an increasing function of $X_{1}(v)$, 
and hence it is also minimized (among possible vertices) 
by the neighbors of $\vs^{1}$ that are not in $\left\{ \vs_{t_{2} - 1}^{2}, \vs_{t_{2}}^{2} \right\}$. 
The third term, $-\log X_{2}(v)$, is not minimized among the neighbors of $\vs^{1}$; 
however, this is a lower order term compared to the first term and it can be seen that the full expression in~\eqref{eq:MLE_odd_odd_ball_notball} is minimized by the neighbors of $\vs^{1}$ that are not in $\left\{ \vs_{t_{2} - 1}^{2}, \vs_{t_{2}}^{2} \right\}$. 
Therefore $\wh{v}_{\ML}$ chooses a neighbor of $\vs^{1}$ that is not in $\left\{ \vs_{t_{2} - 1}^{2}, \vs_{t_{2}}^{2} \right\}$ uniformly at random.

\item \textit{$G_{t_{1}}^{1}$ is a ball, $G_{t_{2}}^{2}$ is not a ball, and $X_{2} \left( \vs^{1} \right) = 1$.} \label{case:odd_odd_ball_notball_5}
This is similar to Case~\eqref{case:odd_odd_ball_notball_4} above and we omit the details for brevity. The MLE is the same: 
$\wh{v}_{\ML}$ chooses one of the $d-1$ neighbors of $\vs^{1}$ that is not in $\left\{ \vs_{t_{2} - 1}^{2}, \vs_{t_{2}}^{2} \right\}$, uniformly at random.

\item \textit{$G_{t_{1}}^{1}$ is a ball, $G_{t_{2}}^{2}$ is not a ball, and $X_{2} \left( \vs^{1} \right) \geq 2$.} \label{case:odd_odd_ball_notball_6}
Let us again consider the three terms in~\eqref{eq:MLE_odd_odd_ball_notball}, 
and let $\vs^{2}$ denote the vertex in the set 
$\left\{ \vs_{t_{2} - 1}^{2}, \vs_{t_{2}}^{2} \right\}$ 
that is closer to $\vs^{1}$. 
The first term, 
$\left( X_{1}(v) + X_{2}(v) \right) \log(d-1)$, 
is minimized on the shortest path between $\vs^{1}$ and $\vs^{2}$. 
The second term is increasing in $X_{1}(v)$, 
while the third term is decreasing in $X_{2}(v)$. 
We also have the constraint that $v \in V_{t_{1}}^{1} \cap V_{t_{2}}^{2} \setminus \left\{ \vs^{1}, \vs_{t_{2} - 1}^{2}, \vs_{t_{2}}^{2} \right\}$. 
Let $P_{12}$ denote the set of vertices that are on the shortest path between $\vs^{1}$ and $\vs^{2}$, excluding $\vs^{1}$ and $\vs^{2}$. 
Define $v'$ to be the vertex in $P_{12} \cap V_{t_{1}}^{1} \cap V_{t_{2}}^{2}$ 
that is closest to $\vs^{1}$ (note that $P_{12} \cap V_{t_{1}}^{1} \cap V_{t_{2}}^{2}$ is nonempty, due to the assumption that $X_{2} \left( \vs^{1} \right) \geq 2$). 
Given the constraints, $v'$ is the unique vertex that minimizes both the first and the second terms in~\eqref{eq:MLE_odd_odd_ball_notball}. Depending on the details, $v'$ might also minimize the third term in~\eqref{eq:MLE_odd_odd_ball_notball}, but even if it does not, it turns out that $v'$ is always the overall minimizer of the expression in~\eqref{eq:MLE_odd_odd_ball_notball}. Therefore $\wh{v}_{\ML} = v'$. 

\end{enumerate}

The case when $G_{t_{1}}^{1}$ is not a ball and $G_{t_{2}}^{2}$ is a ball 
is identical to the above, with the indices $1$ and~$2$ switched. 
Finally, we consider the case when both $G_{t_{1}}^{1}$ and $G_{t_{2}}^{2}$ are not balls. 
We now distinguish four subcases; 
these are based on whether the distance between 
$\left\{ \vs_{t_{1}-1}^{1}, \vs_{t_{1}}^{1} \right\}$ 
and 
$\left\{ \vs_{t_{2} - 1}^{2}, \vs_{t_{2}}^{2} \right\}$ 
is $0$, $1$, or at least $2$, 
and when this distance is $0$, 
we further distinguish between when the size of the intersection of these two sets is $1$ or $2$. 

\begin{enumerate}[(1)]
\setcounter{enumi}{6}
\item \label{case:odd_odd_notball_notball_7} 
\textit{$G_{t_{1}}^{1}$ is not a ball, $G_{t_{2}}^{2}$ is not a ball, and $\left\{ \vs_{t_{1}-1}^{1}, \vs_{t_{1}}^{1} \right\} = \left\{ \vs_{t_{2} - 1}^{2}, \vs_{t_{2}}^{2} \right\}$.} 
In this case for any possible vertex $v$ we have that $\ell := X_{1} (v) = X_{2} (v) \geq 1$. Thus, by~\eqref{eq:MLE_odd_odd_notballs}, the function that we need to minimize is 
$\ell \mapsto 2 \ell \log (d-1) - 2 \log \ell$. 
There are now two cases to distinguish, depending on the value of $d$. 

When $d \geq 4$, this function is minimized when $\ell = 1$. 
Therefore $\wh{v}_{\ML}$ chooses one of the $2d-2$ neighbors of the set 
$\left\{ \vs_{t_{1}-1}^{1}, \vs_{t_{1}}^{1} \right\}$, uniformly at random. 

When $d = 3$, this function is minimized when $\ell \in \left\{ 1, 2 \right\}$. 
Therefore $\wh{v}_{\ML}$ chooses one of the 
$2 \left( d - 1 + (d-1)^{2} \right) = 12$ vertices at distance $1$ or $2$ from the set 
$\left\{ \vs_{t_{1}-1}^{1}, \vs_{t_{1}}^{1} \right\}$, uniformly at random.

\item \label{case:odd_odd_notball_notball_8} 
\textit{$G_{t_{1}}^{1}$ is not a ball, $G_{t_{2}}^{2}$ is not a ball, and $\left| \left\{ \vs_{t_{1}-1}^{1}, \vs_{t_{1}}^{1} \right\} \cap \left\{ \vs_{t_{2} - 1}^{2}, \vs_{t_{2}}^{2} \right\} \right| = 1$.} 
To abbreviate notation, let 
$H := \left\{ \vs_{t_{1}-1}^{1}, \vs_{t_{1}}^{1} \right\} \cup \left\{ \vs_{t_{2} - 1}^{2}, \vs_{t_{2}}^{2} \right\}$ 
and 
$\wt{v} := \left\{ \vs_{t_{1}-1}^{1}, \vs_{t_{1}}^{1} \right\} \cap \left\{ \vs_{t_{2} - 1}^{2}, \vs_{t_{2}}^{2} \right\}$. 
We can categorize the vertices in $V \setminus H$ into three groups. 

There are vertices $v$ for which $\ell := X_{1} (v) = X_{2} (v) - 1 \geq 1$. 
For such vertices the quantity in~\eqref{eq:MLE_odd_odd_notballs} is equal to 
$\left( 2 \ell + 1 \right) \log \left( d - 1 \right) - \log \ell - \log \left( \ell + 1 \right)$. 

There are vertices $v$ for which $\ell := X_{2} (v) = X_{1} (v) - 1 \geq 1$. 
For such vertices the quantity in~\eqref{eq:MLE_odd_odd_notballs} is also equal to 
$\left( 2 \ell + 1 \right) \log \left( d - 1 \right) - \log \ell - \log \left( \ell + 1 \right)$. 

Finally, there are vertices $v$ for which $\ell := X_{1} (v) = X_{2} (v) = \delta \left( v, \wt{v} \right) \geq 1$. 
For such vertices the quantity in~\eqref{eq:MLE_odd_odd_notballs} is equal to 
$2 \ell \log (d - 1) - 2 \log \ell$ (which is the same expression as in Case~\eqref{case:odd_odd_notball_notball_7} above). 

When minimizing these quantities, there are now two cases to distinguish, depending on the value of $d$. When $d \geq 4$, $\wh{v}_{\ML}$ chooses one of the $d-2$ neighbors of $\wt{v}$ not in $H$, uniformly at random. 
When $d = 3$, $\wh{v}_{\ML}$ chooses one of $(d-2) + d(d-1) = 7$ nodes that is at distance $1$ or $2$ from $\wt{v}$ and is not in $H$, uniformly at random.

\item \label{case:odd_odd_notball_notball_9} 
\textit{$G_{t_{1}}^{1}$ is not a ball, $G_{t_{2}}^{2}$ is not a ball, and the distance between $\left\{ \vs_{t_{1}-1}^{1}, \vs_{t_{1}}^{1} \right\}$ and $\left\{ \vs_{t_{2} - 1}^{2}, \vs_{t_{2}}^{2} \right\}$ is equal to $1$.} 
To abbreviate notation, let 
$H := \left\{ \vs_{t_{1}-1}^{1}, \vs_{t_{1}}^{1} \right\} \cup \left\{ \vs_{t_{2} - 1}^{2}, \vs_{t_{2}}^{2} \right\}$, 
let 
$\wt{v}^{1}$ denote the vertex in 
$\left\{ \vs_{t_{1}-1}^{1}, \vs_{t_{1}}^{1} \right\}$ that is closest to 
$\left\{ \vs_{t_{2} - 1}^{2}, \vs_{t_{2}}^{2} \right\}$, 
and define $\wt{v}^{2}$ analogously. 
We can categorize the vertices in $V \setminus H$ into four groups. 

There are vertices $v$ for which $\ell := X_{1} (v) = X_{2} (v) - 2 \geq 1$. 
For such vertices the quantity in~\eqref{eq:MLE_odd_odd_notballs} is equal to 
$(2 \ell + 2) \log (d-1) - \log \ell - \log \left( \ell + 2 \right)$. 

There are vertices $v$ for which $\ell := X_{1} (v) = \delta \left( v, \wt{v}^{1} \right) \geq 1$ and $X_{2} (v) = \delta \left( v, \wt{v}^{2} \right) = \ell + 1$. 
For such vertices the quantity in~\eqref{eq:MLE_odd_odd_notballs} is equal to 
$(2 \ell + 1) \log (d-1) - \log \ell - \log \left( \ell + 1 \right)$. 

There are vertices $v$ for which $\ell := X_{2} (v) = \delta \left( v, \wt{v}^{2} \right) \geq 1$ and $X_{1} (v) = \delta \left( v, \wt{v}^{1} \right) = \ell + 1$. 
For such vertices the quantity in~\eqref{eq:MLE_odd_odd_notballs} is also equal to 
$(2 \ell + 1) \log (d-1) - \log \ell - \log \left( \ell + 1 \right)$. 

Finally, there are vertices $v$ for which $\ell := X_{2} (v) = X_{1} (v) - 2 \geq 1$. 
For such vertices the quantity in~\eqref{eq:MLE_odd_odd_notballs} is equal to 
$(2 \ell + 2) \log (d-1) - \log \ell - \log \left( \ell + 2 \right)$. 

Now note that whenever $d \geq 3$ and $\ell \geq 1$, we always have that 
$(2 \ell + 2) \log (d-1) - \log \ell - \log \left( \ell + 2 \right) > 
(2 \ell + 1) \log (d-1) - \log \ell - \log \left( \ell + 1 \right)$. 
Furthermore, note that the function 
$\ell \mapsto \left(2 \ell + 1\right) \log \left(d-1\right) - \log \ell - \log \left( \ell + 1 \right)$ 
on the domain $\ell \geq 1$ is minimized at $\ell = 1$. 
Putting these together we have that $\wh{v}_{\ML}$ chooses one of the $2(d-2)$ nodes that are a neighbor of $\wt{v}^{1}$ or $\wt{v}^{2}$ and not in $H$, uniformly at random.

\item \label{case:odd_odd_notball_notball_10}
\textit{$G_{t_{1}}^{1}$ is not a ball, $G_{t_{2}}^{2}$ is not a ball, and the distance between $\left\{ \vs_{t_{1}-1}^{1}, \vs_{t_{1}}^{1} \right\}$ and $\left\{ \vs_{t_{2} - 1}^{2}, \vs_{t_{2}}^{2} \right\}$ is at least $2$.} 
Again, let 
$\wt{v}^{1}$ denote the vertex in 
$\left\{ \vs_{t_{1}-1}^{1}, \vs_{t_{1}}^{1} \right\}$ that is closest to 
$\left\{ \vs_{t_{2} - 1}^{2}, \vs_{t_{2}}^{2} \right\}$, 
and define $\wt{v}^{2}$ analogously. 
Let $P_{12}$ denote the set of vertices that are on the shortest path between $\wt{v}^{1}$ and $\wt{v}^{2}$, excluding $\wt{v}^{1}$ and $\wt{v}^{2}$.  

The expression in~\eqref{eq:MLE_odd_odd_notballs} has three terms. The first term, $\left( X_{1} (v) + X_{2} (v) \right) \log \left( d - 1 \right)$, is the main term, and it is minimized on $P_{12}$, taking on the value $\delta \left( \wt{v}^{1}, \wt{v}^{2} \right) \log \left( d - 1 \right)$. The other two terms are decreasing in $X_{1}(v)$ and $X_{2}(v)$, respectively. However, it turns out these are lower order terms and that almost always the whole expression in~\eqref{eq:MLE_odd_odd_notballs} is minimized on $P_{12}$, 
and thus the MLE satisfies 
\[
\wh{v}_{\ML} \in \argmax_{v \in P_{12} \cap V_{t_{1}}^{1} \cap V_{t_{2}}^{2}} X_{1} \left( v \right) X_{2} \left( v \right).
\]

There is only one exception to this: when $d = 3$ and $\delta \left( \wt{v}^{1}, \wt{v}^{2} \right) = 2$. 
In this case let $w$ be the unique vertex that is at distance $1$ from both $\wt{v}^{1}$ and $\wt{v}^{2}$, 
and let $w'$ be the unique vertex that is at distance $2$ from both $\wt{v}^{1}$ and $\wt{v}^{2}$. 
In this case the expression in~\eqref{eq:MLE_odd_odd_notballs} is minimized at $w$ and $w'$, so $\wh{v}_{\ML}$ chooses $w$ or $w'$, uniformly at random. 
\end{enumerate}

Now that we have fully described the MLE, 
we can compute the probability that it is correct. 
We may again condition on whether or not $A_{12}$ holds (see~\eqref{eq:MLE_total_prob}) to obtain that 
\begin{equation}\label{eq:MLE_total_prob_odd_odd}
 \p \left( \wh{v}_{\ML} = v^{*} \right) 
 = \p \left( \wh{v}_{\ML} = v^{*} \, \middle| \, A_{12} \right) \cdot \frac{d-1}{d}
 + \p \left( \wh{v}_{\ML} = v^{*} \, \middle| \, A_{12}^{C} \right) \cdot \frac{1}{d}.
\end{equation}
Note that if $A_{12}$ holds then we must be in Cases~\eqref{case:odd_odd_balls_3},~\eqref{case:odd_odd_ball_notball_6}, or~\eqref{case:odd_odd_notball_notball_10}. 
Consequently, Cases~\eqref{case:odd_odd_balls_1},~\eqref{case:odd_odd_balls_2},~\eqref{case:odd_odd_ball_notball_4}, \eqref{case:odd_odd_ball_notball_5}, \eqref{case:odd_odd_notball_notball_7},~\eqref{case:odd_odd_notball_notball_8}, and~\eqref{case:odd_odd_notball_notball_9} imply that $A_{12}$ does not hold. 

Let us start by computing 
$\p \left( \wh{v}_{\ML} = v^{*} \, \middle| \, A_{12}^{C} \right)$. 
Recall from above that in Cases~\eqref{case:odd_odd_balls_3} and~\eqref{case:odd_odd_ball_notball_6} we have that $\wh{v}_{\ML} \in P_{12}$, 
so if $A_{12}^{C}$ holds, then $\wh{v}_{\ML} \neq v^{*}$. 
In Case~\eqref{case:odd_odd_notball_notball_10} we also have that $\wh{v}_{\ML} \in P_{12}$, with one exception (see above), but even then we have that $\wh{v}_{\ML}$ is ``in between'' $\left\{ \vs_{t_{1}-1}^{1}, \vs_{t_{1}}^{1} \right\}$ and $\left\{ \vs_{t_{2} - 1}^{2}, \vs_{t_{2}}^{2} \right\}$, and so if $A_{12}^{C}$ holds, then $\wh{v}_{\ML} \neq v^{*}$. 
Thus only Cases~\eqref{case:odd_odd_balls_1},~\eqref{case:odd_odd_balls_2},~\eqref{case:odd_odd_ball_notball_4}, \eqref{case:odd_odd_ball_notball_5}, \eqref{case:odd_odd_notball_notball_7},~\eqref{case:odd_odd_notball_notball_8}, and~\eqref{case:odd_odd_notball_notball_9} 
contribute to the probability 
$\p \left( \wh{v}_{\ML} = v^{*} \, \middle| \, A_{12}^{C} \right)$. 
In the following computations we use the expressions~\eqref{eq:alpha_uniform} and~\eqref{eq:p_uniform}. 

First, corresponding to Case~\eqref{case:odd_odd_balls_1} above: 
if $\delta \left( v^{*}, \vs_{t_{1}}^{1} \right) = \delta \left( v^{*}, \vs_{t_{2}}^{2} \right) = 1$, 
then the MLE is correct with probability $1/d$. 
This gives a contribution of 
$\frac{2}{t_{1} - 1} \cdot \frac{t_{1} - 1}{t_{1} + 1} \cdot \frac{2}{t_{2} - 1} \cdot \frac{t_{2} - 1}{t_{2} + 1} \cdot \frac{1}{d} 
= \frac{4}{\left( t_{1} + 1 \right) \left( t_{2} + 1 \right)} \cdot \frac{1}{d}$. 

Second, corresponding to Case~\eqref{case:odd_odd_balls_2} above, we distinguish three subcases based on whether $t_{1} = t_{2}$, $t_{1} < t_{2}$, or $t_{1} > t_{2}$. 
If $t_{1} = t_{2}$, 
then if 
$\delta \left( v^{*}, \vs_{t_{1}}^{1} \right) = 1$, 
$\delta \left( v^{*}, \vs_{t_{2}-1}^{2} \right) = 2$, 
and $\vs_{t_{2} - 1}^{2} = \vs_{t_{2}}^{2}$, 
or if 
$\delta \left( v^{*}, \vs_{t_{2}}^{2} \right) = 1$, 
$\delta \left( v^{*}, \vs_{t_{1}-1}^{1} \right) = 2$, 
and $\vs_{t_{1} - 1}^{1} = \vs_{t_{1}}^{1}$, 
then the MLE is correct with probability 
$1 / (2d-2)$. 
This gives a contribution of 
\[
\left( \frac{2}{t_{1} + 1} \cdot \frac{2}{t_{2} - 1} \cdot \frac{t_{2} - 3}{t_{2} + 1} 
+ \frac{2}{t_{2} + 1} \cdot \frac{2}{t_{1} - 1} \cdot \frac{t_{1} - 3}{t_{1} + 1} \right) 
\frac{1}{2 \left( d - 1 \right)} 
= \frac{4 \left( t - 3 \right)}{\left( t - 1 \right) \left( t + 1 \right)^{2}} \cdot \frac{1}{d-1},
\]
where $t := t_{1} = t_{2}$. 
If $t_{1} < t_{2}$, 
then if 
$\delta \left( v^{*}, \vs_{t_{1}}^{1} \right) = 1$, 
$\delta \left( v^{*}, \vs_{t_{2}-1}^{2} \right) = 2$, 
and $\vs_{t_{2} - 1}^{2} = \vs_{t_{2}}^{2}$, 
then the MLE is correct with probability 
$1/(d-1)$. 
This gives a contribution of 
\[
\frac{2}{t_{1} + 1} \cdot \frac{2}{t_{2} - 1} \cdot \frac{t_{2} - 3}{t_{2} + 1} \cdot \frac{1}{d-1} 
= \frac{4 \left( t_{2} - 3 \right)}{\left( t_{2} - 1 \right) \left( t_{1} + 1 \right) \left( t_{2} + 1 \right)} \cdot \frac{1}{d-1}.
\]
If $t_{1} > t_{2}$, then we have the same contribution as in the display above, with the indices $1$ and $2$ switched. 
Altogether we have obtained that the contribution in every subcase is 
\[
\frac{4 \left( \max \left\{ t_{1}, t_{2} \right\} - 3 \right)}{\left( \max \left\{ t_{1}, t_{2} \right\} - 1 \right) \left( t_{1} + 1 \right) \left( t_{2} + 1 \right)} \cdot \frac{1}{d-1} 
< 
\frac{4}{\left( t_{1} + 1 \right) \left( t_{2} + 1 \right)} \cdot \frac{1}{d-1}.
\]

Next, we turn to Case~\eqref{case:odd_odd_ball_notball_4} above. 
We have that if 
$\delta \left( v^{*}, \vs_{t_{1}}^{1} \right) = 1$, 
$\delta \left( v^{*}, \vs_{t_{2} - 1}^{2} \right) = 1$, 
and $\delta \left( v^{*}, \vs_{t_{2}}^{2} \right) = 2$, 
then the MLE is correct with probability $1/(d-1)$. 
This gives a contribution of 
$\frac{2}{t_{1} + 1} \cdot \frac{2}{t_{2} - 1} \cdot \frac{2}{t_{2} + 1} \cdot \frac{1}{d-1}$. 
There is an analogous term coming from when $G_{t_{1}}^{1}$ is not a ball and $G_{t_{2}}^{2}$ is a ball, 
giving a contribution of 
$\frac{2}{t_{2} + 1} \cdot \frac{2}{t_{1} - 1} \cdot \frac{2}{t_{1} + 1} \cdot \frac{1}{d-1}$. 
The combined contribution from the two terms is 
\[
\frac{4}{\left( t_{1} + 1 \right) \left( t_{2} + 1 \right)} \cdot \frac{1}{d - 1} \left( \frac{2}{t_{1} - 1} + \frac{2}{t_{2} - 1} \right) 
\leq \frac{4}{\left( t_{1} + 1 \right) \left( t_{2} + 1 \right)} \cdot \frac{1}{d - 1},
\]
where the inequality follows from the assumption that $\min \left\{ t_{1}, t_{2} \right\} \geq 5$.

Next, we turn to Case~\eqref{case:odd_odd_ball_notball_5} above. 
We have that if 
$\delta \left( v^{*}, \vs_{t_{1}}^{1} \right) = 1$, 
$\delta \left( v^{*}, \vs_{t_{2} - 1}^{2} \right) = 2$, 
and $\delta \left( v^{*}, \vs_{t_{2}}^{2} \right) = 3$, 
then the MLE is correct with probability $1/(d-1)$. 
This gives a contribution of 
$\frac{2}{t_{1} + 1} \cdot \frac{2}{t_{2} - 1} \cdot \frac{4}{t_{2} + 1} \cdot \frac{1}{d-1}$. 
There is an analogous term coming from when $G_{t_{1}}^{1}$ is not a ball and $G_{t_{2}}^{2}$ is a ball, 
giving a contribution of 
$\frac{2}{t_{2} + 1} \cdot \frac{2}{t_{1} - 1} \cdot \frac{4}{t_{1} + 1} \cdot \frac{1}{d-1}$. 
The combined contribution from the two terms is 
\[
\frac{4}{\left( t_{1} + 1 \right) \left( t_{2} + 1 \right)} \cdot \frac{1}{d - 1} \left( \frac{4}{t_{1} - 1} + \frac{4}{t_{2} - 1} \right) 
\leq \frac{4}{\left( t_{1} + 1 \right) \left( t_{2} + 1 \right)} \cdot \frac{2}{d - 1},
\]
where the inequality follows from the assumption that $\min \left\{ t_{1}, t_{2} \right\} \geq 5$. 

Next, we turn to Case~\eqref{case:odd_odd_notball_notball_7} above, where we have to distinguish between $d \geq 4$ and $d = 3$. 
When $d \geq 4$, 
then if 
$\delta \left( v^{*}, \vs_{t_{1} - 1}^{1} \right) = 1$, 
$\delta \left( v^{*}, \vs_{t_{1}}^{1} \right) = 2$, 
$\delta \left( v^{*}, \vs_{t_{2} - 1}^{2} \right) = 1$,  
$\delta \left( v^{*}, \vs_{t_{2}}^{2} \right) = 2$, 
and $\vs_{t_{1}}^{1} = \vs_{t_{2}}^{2}$, 
then the MLE is correct with probability $1/(2d-2)$. 
This gives a contribution of 
\begin{align}
\frac{2}{t_{1} - 1} \cdot \frac{2}{t_{1} + 1} \cdot \frac{2}{t_{2} - 1} \cdot \frac{2}{t_{2} + 1} \cdot \frac{1}{d-1} \cdot \frac{1}{2(d-1)} 
&= \frac{4}{\left( t_{1} + 1 \right) \left( t_{2} + 1 \right)} \cdot \frac{1}{(d-1)^{2}} \cdot \frac{2}{\left( t_{1} - 1 \right) \left( t_{2} - 1 \right)} \label{eq:odd_odd_case7} \\
&\leq \frac{1}{8} \cdot \frac{4}{\left( t_{1} + 1 \right) \left( t_{2} + 1 \right)} \cdot \frac{1}{(d-1)^{2}}, \notag
\end{align}
where the inequality follows from the assumption that $\min \left\{ t_{1}, t_{2} \right\} \geq 5$. 
When $d = 3$, there are now two situations when the MLE has a chance of being correct:  
\begin{itemize} 
\item if 
$\delta \left( v^{*}, \vs_{t_{1} - 1}^{1} \right) = 1$, 
$\delta \left( v^{*}, \vs_{t_{1}}^{1} \right) = 2$, 
$\delta \left( v^{*}, \vs_{t_{2} - 1}^{2} \right) = 1$,  
$\delta \left( v^{*}, \vs_{t_{2}}^{2} \right) = 2$, 
and $\vs_{t_{1}}^{1} = \vs_{t_{2}}^{2}$, 
\item or if 
$\delta \left( v^{*}, \vs_{t_{1} - 1}^{1} \right) = 2$, 
$\delta \left( v^{*}, \vs_{t_{1}}^{1} \right) = 3$, 
$\delta \left( v^{*}, \vs_{t_{2} - 1}^{2} \right) = 2$,  
$\delta \left( v^{*}, \vs_{t_{2}}^{2} \right) = 3$, 
$\vs_{t_{1} - 1}^{1} = \vs_{t_{2} - 1}^{2}$, 
and $\vs_{t_{1}}^{1} = \vs_{t_{2}}^{2}$; 
\end{itemize}
in both cases the MLE is correct with probability $1/12$. 
This gives a contribution of 
\begin{multline*}
\left( \frac{2}{t_{1} - 1} \cdot \frac{2}{t_{1} + 1} \cdot \frac{2}{t_{2} - 1} \cdot \frac{2}{t_{2} + 1} \cdot \frac{1}{2} 
+ \frac{2}{t_{1} - 1} \cdot \frac{4}{t_{1} + 1} \cdot \frac{2}{t_{2} - 1} \cdot \frac{4}{t_{2} + 1} \cdot \frac{1}{2} \cdot \frac{1}{2} \right) \cdot \frac{1}{12} \\
= \frac{2}{\left( t_{1} + 1 \right) \left( t_{2} + 1 \right) \left( t_{1} - 1 \right) \left( t_{2} - 1 \right)}.
\end{multline*}
This is equal to the quantity in~\eqref{eq:odd_odd_case7}, when $d = 3$ is substituted. Thus no matter what $d \geq 3$ is, the contribution is always at most 
\[
\frac{1}{8} \cdot \frac{4}{\left( t_{1} + 1 \right) \left( t_{2} + 1 \right)} \cdot \frac{1}{(d-1)^{2}}.
\]

Next, we turn to Case~\eqref{case:odd_odd_notball_notball_8} above, where we again distinguish between $d \geq 4$ and $d = 3$. 
When $d \geq 4$, 
then if 
$\delta \left( v^{*}, \vs_{t_{1} - 1}^{1} \right) = 1$, 
$\delta \left( v^{*}, \vs_{t_{1}}^{1} \right) = 2$, 
$\delta \left( v^{*}, \vs_{t_{2} - 1}^{2} \right) = 1$,  
$\delta \left( v^{*}, \vs_{t_{2}}^{2} \right) = 2$, 
and $\vs_{t_{1}}^{1} \neq \vs_{t_{2}}^{2}$, 
then the MLE is correct with probability $1/(d-2)$. 
This gives a contribution of 
\begin{align}
\frac{2}{t_{1} - 1} \cdot \frac{2}{t_{1} + 1} \cdot \frac{2}{t_{2} - 1} \cdot \frac{2}{t_{2} + 1} \cdot \frac{d-2}{d-1} \cdot \frac{1}{d-2} 
&= \frac{4}{\left( t_{1} + 1 \right) \left( t_{2} + 1 \right)} \cdot \frac{1}{d-1} \cdot \frac{4}{\left( t_{1} - 1 \right) \left( t_{2} - 1 \right)} \label{eq:odd_odd_case8} \\
&\leq \frac{1}{4} \cdot \frac{4}{\left( t_{1} + 1 \right) \left( t_{2} + 1 \right)} \cdot \frac{1}{d-1}, \notag
\end{align}
where the inequality follows from the assumption that $\min \left\{ t_{1}, t_{2} \right\} \geq 5$. 
When $d=3$, there are now four situations when the MLE has a chance of being correct: 
\begin{itemize}
\item if 
$\delta \left( v^{*}, \vs_{t_{1} - 1}^{1} \right) = 1$, 
$\delta \left( v^{*}, \vs_{t_{1}}^{1} \right) = 2$, 
$\delta \left( v^{*}, \vs_{t_{2} - 1}^{2} \right) = 1$,  
$\delta \left( v^{*}, \vs_{t_{2}}^{2} \right) = 2$, 
and $\vs_{t_{1}}^{1} \neq \vs_{t_{2}}^{2}$; 
\item or if 
$\delta \left( v^{*}, \vs_{t_{1} - 1}^{1} \right) = 2$, 
$\delta \left( v^{*}, \vs_{t_{1}}^{1} \right) = 3$, 
$\delta \left( v^{*}, \vs_{t_{2} - 1}^{2} \right) = 2$,  
$\delta \left( v^{*}, \vs_{t_{2}}^{2} \right) = 3$, 
$\vs_{t_{1} - 1}^{1} = \vs_{t_{2} - 1}^{2}$, 
and $\vs_{t_{1}}^{1} \neq \vs_{t_{2}}^{2}$; 
\item or if 
$\delta \left( v^{*}, \vs_{t_{1} - 1}^{1} \right) = 1$, 
$\delta \left( v^{*}, \vs_{t_{1}}^{1} \right) = 2$, 
$\delta \left( v^{*}, \vs_{t_{2} - 1}^{2} \right) = 2$,  
$\delta \left( v^{*}, \vs_{t_{2}}^{2} \right) = 3$, 
and $\vs_{t_{1}}^{1} = \vs_{t_{2} - 1}^{2}$;  
\item or if 
$\delta \left( v^{*}, \vs_{t_{1} - 1}^{1} \right) = 2$, 
$\delta \left( v^{*}, \vs_{t_{1}}^{1} \right) = 3$, 
$\delta \left( v^{*}, \vs_{t_{2} - 1}^{2} \right) = 1$,  
$\delta \left( v^{*}, \vs_{t_{2}}^{2} \right) = 2$, 
and $\vs_{t_{1} - 1}^{1} = \vs_{t_{2}}^{2}$; 
\end{itemize}
in all four cases the MLE is correct with probability $1/7$. 
This gives a contribution of 
\begin{align*}
&\frac{2}{t_{1} - 1} \cdot \frac{2}{t_{1} + 1} \cdot \frac{2}{t_{2} - 1} \cdot \frac{2}{t_{2} + 1} \cdot \frac{1}{2} \cdot \frac{1}{7} 
+ \frac{2}{t_{1} - 1} \cdot \frac{4}{t_{1} + 1} \cdot \frac{2}{t_{2} - 1} \cdot \frac{4}{t_{2} + 1} \cdot \frac{1}{2} \cdot \frac{1}{2} \cdot \frac{1}{7}  \\
&+ 
\frac{2}{t_{1} - 1} \cdot \frac{2}{t_{1} + 1} \cdot \frac{2}{t_{2} - 1} \cdot \frac{4}{t_{2} + 1} \cdot \frac{1}{2} \cdot \frac{1}{7} 
+ \frac{2}{t_{1} - 1} \cdot \frac{4}{t_{1} + 1} \cdot \frac{2}{t_{2} - 1} \cdot \frac{2}{t_{2} + 1} \cdot \frac{1}{2} \cdot \frac{1}{7}   \\
&= \frac{8}{\left( t_{1} + 1 \right) \left( t_{2} + 1 \right) \left( t_{1} - 1 \right) \left( t_{2} - 1 \right)}.
\end{align*}
This is equal to the quantity in~\eqref{eq:odd_odd_case8}, when $d=3$ is substituted. Thus no matter what $d \geq 3$ is, the contribution is always at most 
\[
\frac{1}{4} \cdot \frac{4}{\left( t_{1} + 1 \right) \left( t_{2} + 1 \right)} \cdot \frac{1}{d-1}.
\]

Finally, we turn to Case~\eqref{case:odd_odd_notball_notball_9} above. 
There are now two situations when the MLE has a chance of being correct: 
\begin{itemize}
\item if 
$\delta \left( v^{*}, \vs_{t_{1} - 1}^{1} \right) = 1$, 
$\delta \left( v^{*}, \vs_{t_{1}}^{1} \right) = 2$, 
$\delta \left( v^{*}, \vs_{t_{2} - 1}^{2} \right) = 2$,  
$\delta \left( v^{*}, \vs_{t_{2}}^{2} \right) = 3$, 
and $\vs_{t_{1}}^{1} \neq \vs_{t_{2} - 1}^{2}$; 
\item or if 
$\delta \left( v^{*}, \vs_{t_{1} - 1}^{1} \right) = 2$, 
$\delta \left( v^{*}, \vs_{t_{1}}^{1} \right) = 3$, 
$\delta \left( v^{*}, \vs_{t_{2} - 1}^{2} \right) = 1$,  
$\delta \left( v^{*}, \vs_{t_{2}}^{2} \right) = 2$, 
and $\vs_{t_{1} - 1}^{1} \neq \vs_{t_{2}}^{2}$; 
\end{itemize} 
in both cases the MLE is correct with probability $1/(2d-4)$. This gives a contribution of 
\begin{multline*}
\frac{2}{t_{1} - 1} \cdot \frac{2}{t_{1} + 1} \cdot \frac{2}{t_{2} - 1} \cdot \frac{4}{t_{2} + 1} \cdot \frac{d-2}{d-1} \cdot \frac{1}{2(d-2)} 
+ \frac{2}{t_{1} - 1} \cdot \frac{4}{t_{1} + 1} \cdot \frac{2}{t_{2} - 1} \cdot \frac{2}{t_{2} + 1} \cdot \frac{d-2}{d-1} \cdot \frac{1}{2(d-2)} \\
= \frac{4}{(t_{1} + 1) (t_{2} + 1)} \cdot \frac{1}{d-1} \cdot \frac{8}{(t_{1} - 1) (t_{2} - 1)}
\leq \frac{1}{2} \cdot \frac{4}{(t_{1} + 1) (t_{2} + 1)} \cdot \frac{1}{d-1},
\end{multline*}
where the inequality follows from the assumption that $\min \left\{ t_{1}, t_{2} \right\} \geq 5$. 

In summary, putting the contributions from all these cases together we are now ready to bound the probability 
$\p \left( \wh{v}_{\ML} = v^{*} \, \middle| \, A_{12}^{C} \right)$. We have that 
\begin{align}
\p \left( \wh{v}_{\ML} = v^{*} \, \middle| \, A_{12}^{C} \right) 
&\leq 
\frac{4}{\left( t_{1} + 1 \right) \left( t_{2} + 1 \right)} 
\left\{ \frac{1}{d} + \frac{1}{d-1} \left( 1 + 1 + 2 + \frac{1}{8(d-1)} + \frac{1}{4} + \frac{1}{2} \right) \right\} \notag\\
&\leq \frac{263}{24} \cdot \frac{1}{\left( t_{1} + 1 \right) \left( t_{2} + 1 \right)} 
\leq \frac{11}{\left( t_{1} + 1 \right) \left( t_{2} + 1 \right)}, \label{eq:probMLEcorrectA12C_simple}
\end{align}
where in the second inequality we used that $d \geq 3$.

We now turn to computing 
$\p \left( \wh{v}_{\ML} = v^{*} \, \middle| \, A_{12} \right)$. 
To do this, we condition on the values of 
$\delta \left( v^{*}, \vs_{t_{1} - 1}^{1} \right)$ 
and 
$\delta \left( v^{*}, \vs_{t_{2} - 1}^{2} \right)$. 
To this end, define for 
$1 \leq h_{1} \leq \left( t_{1} - 1 \right) / 2$ 
and 
$1 \leq h_{2} \leq \left( t_{2} - 1 \right) / 2$ 
the events 
\begin{align*}
D_{1} \left( h_{1} \right) 
&:= \left\{ \delta \left( v^{*}, \vs_{t_{1} - 1}^{1} \right) = h_{1} \right\}, \\
D_{2} \left( h_{2} \right) 
&:= \left\{ \delta \left( v^{*}, \vs_{t_{2} - 1}^{2} \right) = h_{2} \right\}. 
\end{align*}
The events $D_{1} \left( h_{1} \right)$ and $D_{2} \left( h_{2} \right)$ are independent 
and they are also independent of $A_{12}$, 
so 
\[
\p \left( D_{1} \left( h_{1} \right) \cap D_{2} \left( h_{2} \right) \, \middle| \, A_{12} \right) 
= \p \left( D_{1} \left( h_{1} \right) \cap D_{2} \left( h_{2} \right) \right) 
= \p \left( D_{1} \left( h_{1} \right) \right) \p \left( D_{2} \left( h_{2} \right) \right) 
= \frac{2}{t_{1} - 1} \cdot \frac{2}{t_{2} - 1}.
\]
Thus conditioning on the values of 
$\delta \left( v^{*}, \vs_{t_{1} - 1}^{1} \right)$ 
and 
$\delta \left( v^{*}, \vs_{t_{2} - 1}^{2} \right)$ 
we have that 
\begin{equation}\label{eq:condh1h2}
\p \left( \wh{v}_{\ML} = v^{*} \, \middle| \, A_{12} \right) 
= \frac{4}{\left( t_{1} - 1 \right) \left( t_{2} - 1 \right)} 
\sum_{h_{1} = 1}^{(t_{1} - 1) / 2} \sum_{h_{2} = 1}^{(t_{2} - 1) / 2} 
\p \left( \wh{v}_{\ML} = v^{*} \, \middle| \, A_{12} \cap D_{1} \left( h_{1} \right) \cap D_{2} \left( h_{2} \right) \right).
\end{equation}
So now what remains is to compute 
$\p \left( \wh{v}_{\ML} = v^{*} \, \middle| \, A_{12} \cap D_{1} \left( h_{1} \right) \cap D_{2} \left( h_{2} \right) \right)$ 
and to sum this over possible values of $h_{1}$ and $h_{2}$. 
To do this we further define the events 
\begin{align*}
B_{1} &:= \left\{ \vs_{t_{1} - 1}^{1} = \vs_{t_{1}}^{1} \right\}, \\
B_{2} &:= \left\{ \vs_{t_{2} - 1}^{2} = \vs_{t_{2}}^{2} \right\}. 
\end{align*}
In words, $B_{1}$ and $B_{2}$ are the events that the virtual source stays in place at the last time step in the first and the second sample, respectively. 
Note that $B_{1}$ and $B_{2}$ are independent, even conditioned on 
$A_{12} \cap D_{1} \left( h_{1} \right) \cap D_{2} \left( h_{2} \right)$. 
Therefore 
\[
\p \left( B_{1} \cap B_{2} \, \middle| \, A_{12} \cap D_{1} \left( h_{1} \right) \cap D_{2} \left( h_{2} \right) \right) 
= \p \left( B_{1} \, \middle| \, D_{1} \left( h_{1} \right) \right) \p \left( B_{2} \, \middle| \, D_{2} \left( h_{2} \right) \right) 
= \frac{t_{1} + 1 - 2h_{1}}{t_{1} + 1} \cdot \frac{t_{2} + 1 - 2h_{2}}{t_{2} + 1},
\]
and similarly
\begin{align*}
\p \left( B_{1}^{C} \cap B_{2} \, \middle| \, A_{12} \cap D_{1} \left( h_{1} \right) \cap D_{2} \left( h_{2} \right) \right) 
&= \frac{2h_{1}}{t_{1} + 1} \cdot \frac{t_{2} + 1 - 2h_{2}}{t_{2} + 1}, \\
\p \left( B_{1} \cap B_{2}^{C} \, \middle| \, A_{12} \cap D_{1} \left( h_{1} \right) \cap D_{2} \left( h_{2} \right) \right) 
&= \frac{t_{1} + 1 - 2h_{1}}{t_{1} + 1} \cdot \frac{2h_{2}}{t_{2} + 1}, \\
\p \left( B_{1}^{C} \cap B_{2}^{C} \, \middle| \, A_{12} \cap D_{1} \left( h_{1} \right) \cap D_{2} \left( h_{2} \right) \right) 
&= \frac{2h_{1}}{t_{1} + 1} \cdot \frac{2h_{2}}{t_{2} + 1}.
\end{align*} 
By conditioning on whether or not the events $B_{1}$ and $B_{2}$ hold, we may break up the probability 
$\p \left( \wh{v}_{\ML} = v^{*} \, \middle| \, A_{12} \cap D_{1} \left( h_{1} \right) \cap D_{2} \left( h_{2} \right) \right)$ 
into a sum with four terms: 
\begin{align}
&\p \left( \wh{v}_{\ML} = v^{*} \, \middle| \, A_{12} \cap D_{1} \left( h_{1} \right) \cap D_{2} \left( h_{2} \right) \right) \notag \\
&\qquad = \frac{\left( t_{1} + 1 - 2h_{1} \right) \left( t_{2} + 1 - 2h_{2} \right)}{\left( t_{1} + 1 \right) \left( t_{2} + 1 \right)} 
\p \left( \wh{v}_{\ML} = v^{*} \, \middle| \, A_{12} \cap D_{1} \left( h_{1} \right) \cap D_{2} \left( h_{2} \right) \cap B_{1} \cap B_{2} \right) \label{eq:sumB1B2} \\
&\qquad\quad + \frac{\left( 2h_{1} \right) \left( t_{2} + 1 - 2h_{2} \right)}{\left( t_{1} + 1 \right) \left( t_{2} + 1 \right)} 
\p \left( \wh{v}_{\ML} = v^{*} \, \middle| \, A_{12} \cap D_{1} \left( h_{1} \right) \cap D_{2} \left( h_{2} \right) \cap B_{1}^{C} \cap B_{2} \right) \label{eq:sumB1CB2} \\
&\qquad\quad + \frac{\left( t_{1} + 1 - 2h_{1} \right) \left( 2h_{2} \right)}{\left( t_{1} + 1 \right) \left( t_{2} + 1 \right)} 
\p \left( \wh{v}_{\ML} = v^{*} \, \middle| \, A_{12} \cap D_{1} \left( h_{1} \right) \cap D_{2} \left( h_{2} \right) \cap B_{1} \cap B_{2}^{C} \right) \label{eq:sumB1B2C} \\
&\qquad\quad + \frac{4 h_{1} h_{2}}{\left( t_{1} + 1 \right) \left( t_{2} + 1 \right)} 
\p \left( \wh{v}_{\ML} = v^{*} \, \middle| \, A_{12} \cap D_{1} \left( h_{1} \right) \cap D_{2} \left( h_{2} \right) \cap B_{1}^{C} \cap B_{2}^{C} \right) \label{eq:sumB1CB2C}
\end{align}
We now compute each of these four conditional probabilities in turn. 

We start with~\eqref{eq:sumB1CB2} and~\eqref{eq:sumB1B2C}, as these are the simplest cases among the four. Given the event 
$A_{12} \cap D_{1} \left( h_{1} \right) \cap D_{2} \left( h_{2} \right) \cap B_{1} \cap B_{2}^{C}$ 
in~\eqref{eq:sumB1B2C}, 
Case~\eqref{case:odd_odd_ball_notball_6} describes the MLE. 
Specifically, if $v'$ denotes the vertex in 
$P_{12} \cap V_{t_{1}}^{1} \cap V_{t_{2}}^{2}$ 
that is closest to $\vs^{1}$, 
then we have that $\wh{v}_{\ML} = v'$. 
Therefore in this case 
$\wh{v}_{\ML} = v^{*}$ if and only if $h_{1} = 1$ or $h_{2} = \left( t_{2} - 1 \right) / 2$. 
Thus we have that 
\begin{multline*}
\p \left( \wh{v}_{\ML} = v^{*} \, \middle| \, A_{12} \cap D_{1} \left( h_{1} \right) \cap D_{2} \left( h_{2} \right) \cap B_{1} \cap B_{2}^{C} \right) 
= \mathbf{1}_{\left\{ h_{1} = 1 \right\} \cup \left\{ h_{2} = \left( t_{2} - 1 \right) / 2 \right\}} \\
= \mathbf{1}_{\left\{ h_{1} = 1 \right\}} + \mathbf{1}_{\left\{ h_{2} = \left( t_{2} - 1 \right) / 2 \right\}} - \mathbf{1}_{\left\{ h_{1} = 1 \right\} \cap \left\{ h_{2} = \left( t_{2} - 1 \right) / 2 \right\}}.
\end{multline*}
Plugging this back into~\eqref{eq:sumB1B2C} and summing over $h_{1}$ and $h_{2}$ we obtain that 
\begin{align*}
&\sum_{h_{1} = 1}^{(t_{1} - 1) / 2} \sum_{h_{2} = 1}^{(t_{2} - 1) / 2} 
\frac{\left( t_{1} + 1 - 2h_{1} \right) \left( 2h_{2} \right)}{\left( t_{1} + 1 \right) \left( t_{2} + 1 \right)} 
\p \left( \wh{v}_{\ML} = v^{*} \, \middle| \, A_{12} \cap D_{1} \left( h_{1} \right) \cap D_{2} \left( h_{2} \right) \cap B_{1} \cap B_{2}^{C} \right) \\
&\qquad = \frac{4}{\left( t_{1} + 1 \right) \left( t_{2} + 1 \right)} 
\sum_{h_{1} = 1}^{(t_{1} - 1) / 2} \sum_{h_{2} = 1}^{(t_{2} - 1) / 2} 
\left( \frac{t_{1} + 1}{2} - h_{1} \right) h_{2} 
\mathbf{1}_{\left\{ h_{1} = 1 \right\} \cup \left\{ h_{2} = \left( t_{2} - 1 \right) / 2 \right\}} \\
&\qquad 
= \frac{4}{\left( t_{1} + 1 \right) \left( t_{2} + 1 \right)} \cdot \frac{t_{1} - 1}{2} \sum_{h_{2} = 1}^{(t_{2} - 1) / 2} h_{2} 
+ \frac{4}{\left( t_{1} + 1 \right) \left( t_{2} + 1 \right)} \cdot \frac{t_{2} - 1}{2} \sum_{h_{1} = 1}^{(t_{1} - 1) / 2} \left( \frac{t_{1} + 1}{2} - h_{1} \right) \\
&\qquad\quad - \frac{4}{\left( t_{1} + 1 \right) \left( t_{2} + 1 \right)} \cdot \frac{t_{1} - 1}{2} \cdot \frac{t_{2} - 1}{2} \\
&\qquad 
= \frac{\left( t_{1} - 1 \right) \left( t_{2} - 1 \right)}{4 \left( t_{1} + 1 \right)} 
+ \frac{\left( t_{1} - 1 \right) \left( t_{2} - 1 \right)}{4 \left( t_{2} + 1 \right)} 
- \frac{\left( t_{1} - 1 \right) \left( t_{2} - 1 \right)}{\left( t_{1} + 1 \right) \left( t_{2} + 1 \right)}. 
\end{align*}
Multiplying this expression by 
$4 / \left\{ \left( t_{1} - 1 \right) \left( t_{2} - 1 \right) \right\}$ 
we thus see that the contribution to~\eqref{eq:condh1h2} from~\eqref{eq:sumB1B2C} is 
\begin{equation}\label{eq:contributionB1B2C}
\frac{1}{t_{1} + 1} + \frac{1}{t_{2} + 1} - \frac{4}{\left( t_{1} + 1 \right) \left( t_{2} + 1 \right)}.
\end{equation}
Observe that~\eqref{eq:sumB1CB2} is analogous to~\eqref{eq:sumB1B2C} with the two samples switched. 
Since the expression in~\eqref{eq:contributionB1B2C} is symmetric with respect to $t_{1}$ and $t_{2}$, this means that the 
contribution to~\eqref{eq:condh1h2} from~\eqref{eq:sumB1CB2} is 
also equal to the expression in~\eqref{eq:contributionB1B2C}.

We now turn to the expression in~\eqref{eq:sumB1CB2C}. 
Given the event 
$A_{12} \cap D_{1} \left( h_{1} \right) \cap D_{2} \left( h_{2} \right) \cap B_{1}^{C} \cap B_{2}^{C}$, 
Case~\eqref{case:odd_odd_notball_notball_10} 
describes the MLE. 
In particular, note that on this event we have that 
$X_{1} (v) + X_{2} (v) = h_{1} + h_{2}$ for every $v \in P_{12}$, 
and thus 
$X_{1}(v) X_{2} (v) = X_{1}(v) \left( h_{1} + h_{2} - X_{1}(v) \right)$. 
Letting $d := X_{1}(v)$ to abbreviate notation, 
note that the function 
$d \mapsto d \left( h_{1} + h_{2} - d \right)$ 
is a quadratic function that 
is maximized (among integers) at $d = \left( h_{1} + h_{2} \right) / 2$ if $h_{1} + h_{2}$ is even, 
and at $d = \left( h_{1} + h_{2} - 1\right) / 2$ and $d = \left( h_{1} + h_{2} + 1\right) / 2$ if $h_{1} + h_{2}$ is odd. 
This allows us to understand the MLE and the expression in~\eqref{eq:sumB1CB2C} as follows: 
\begin{itemize}
\item If $h_{1} = h_{2} = 1$, then if $d = 3$ then the MLE is correct with probability $1/2$, and if $d \geq 4$ then the MLE is correct with probability $1$. In either case we can bound this probability by $1$.

\item If $2 < h_{1} + h_{2} \leq \min \left\{ t_{1}, t_{2} \right\} - 1$ and $h_{1} + h_{2}$ is even, then $\wh{v}_{\ML}$ is the unique vertex $v$ such that 
$X_{1} (v) = X_{2} (v) = (h_{1} + h_{2}) / 2$. 
Thus in this case $\wh{v}_{\ML} = v^{*}$ if and only if $h_{1} = h_{2}$. 

\item If $2 < h_{1} + h_{2} \leq \min \left\{ t_{1}, t_{2} \right\} - 1$ and $h_{1} + h_{2}$ is odd, 
then $\wh{v}_{\ML}$ picks uniformly at random among the two vertices 
for which $\left\{ X_{1} (v), X_{2} (v) \right\} = \left\{ (h_{1} + h_{2} - 1) / 2, (h_{1} + h_{2} + 1) / 2 \right\}$. 
Thus in this case the MLE is correct with probability $1/2$ if $\left| h_{1} - h_{2} \right| = 1$, and not correct otherwise. 

\item If $h_{1} + h_{2} \geq \min \left\{ t_{1}, t_{2} \right\}$ (note that this can only occur if $t_{1} \neq t_{2}$), then 
if $v \in P_{12}$ is such that $X_{1}(v) \in \left\{ \left\lfloor (h_{1} + h_{2}) / 2 \right\rfloor, \left\lceil (h_{1} + h_{2}) / 2 \right\rceil \right\}$, 
then $v \notin V_{t_{1}}^{1} \cap V_{t_{2}}^{2}$.  
Therefore, since $d \mapsto d \left( h_{1} + h_{2} - d \right)$ is a quadratic function, 
$\wh{v}_{\ML}$ is the unique vertex $v \in P_{12} \cap V_{t_{1}}^{1} \cap V_{t_{2}}^{2}$ such that $X_{1}(v)$ is closest to 
$\left\{ \left\lfloor (h_{1} + h_{2}) / 2 \right\rfloor, \left\lceil (h_{1} + h_{2}) / 2 \right\rceil \right\}$. 
To understand this better, assume (without loss of generality) that $t_{1} \leq t_{2}$. 
Then we have that $\wh{v}_{\ML} = v^{*}$ if and only if $h_{1} = (t_{1} - 1) / 2$. 

\end{itemize}
Altogether we have obtained, assuming $t_{1} \leq t_{2}$, that 
\begin{multline*}
\p \left( \wh{v}_{\ML} = v^{*} \, \middle| \, A_{12} \cap D_{1} \left( h_{1} \right) \cap D_{2} \left( h_{2} \right) \cap B_{1}^{C} \cap B_{2}^{C} \right) \\
\leq 
\mathbf{1}_{\left\{ h_{1} = h_{2} \right\}} 
+ \frac{1}{2} \mathbf{1}_{\left\{ \left| h_{1} - h_{2} \right| = 1, h_{1} + h_{2} \leq t_{1} - 1 \right\}} 
+ \mathbf{1}_{\left\{ h_{1} = (t_{1} - 1) / 2, h_{1} + h_{2} \geq t_{1} \right\}}.
\end{multline*}
Multiplying the right hand side by $h_{1} h_{2}$ and summing over $h_{1}$ and $h_{2}$, we obtain that 
\begin{multline*}
\sum_{h_{1} = 1}^{(t_{1} - 1) / 2} \sum_{h_{2} = 1}^{(t_{2} - 1) / 2} 
h_{1} h_{2} \left( \mathbf{1}_{\left\{ h_{1} = h_{2} \right\}} 
+ \frac{1}{2} \mathbf{1}_{\left\{ \left| h_{1} - h_{2} \right| = 1, h_{1} + h_{2} \leq t_{1} - 1 \right\}} 
+ \mathbf{1}_{\left\{ h_{1} = (t_{1} - 1) / 2, h_{1} + h_{2} \geq t_{1} \right\}} \right) \\
\begin{aligned}
&= 
\sum_{h = 1}^{(t_{1} - 1) / 2} h^{2} 
+ \sum_{h = 1}^{(t_{1} - 3) / 2} h (h + 1) 
+ \frac{t_{1} - 1}{2} \sum_{h_{2} = (t_{1} + 1) / 2}^{(t_{2} - 1) / 2} h_{2} \qquad \\
&= 
\frac{\left( t_{1} - 1 \right) \left[ \left( t_{1} - 3 \right) \left( t_{1} + 1 \right) + 3 \left( t_{2} - 1 \right) \left( t_{2} + 1 \right) \right]}{48}. \qquad
\end{aligned}
\end{multline*}
Multiplying this by 
$16 / \left\{ \left( t_{1} - 1 \right) \left( t_{2} - 1 \right) \left( t_{1} + 1 \right) \left( t_{2} + 1 \right) \right\}$ 
we thus see that the contribution to~\eqref{eq:condh1h2} from~\eqref{eq:sumB1CB2C} is at most
\[
\frac{\left( t_{1} - 3 \right) \left( t_{1} + 1 \right) + 3 \left( t_{2} - 1 \right) \left( t_{2} + 1 \right)}{3 \left( t_{2} - 1 \right) \left( t_{1} + 1 \right) \left( t_{2} + 1 \right)}
= \frac{1}{t_{1} + 1} + \frac{t_{1} - 3}{3\left( t_{2} - 1 \right) \left( t_{2} + 1 \right)} 
\leq \frac{1}{t_{1} + 1} + \frac{1/3}{t_{2} + 1}.
\]
Recall that here we assumed that $t_{1} \leq t_{2}$, 
so in general the contribution to~\eqref{eq:condh1h2} from~\eqref{eq:sumB1CB2C} is at most 
\begin{equation}\label{eq:contribution_B1CB2C}
\frac{1}{\min \left\{ t_{1}, t_{2} \right\} + 1} + \frac{1/3}{\max \left\{ t_{1}, t_{2} \right\} + 1}.
\end{equation}

The expression in~\eqref{eq:sumB1B2} is similar to that in~\eqref{eq:sumB1CB2C}; 
in fact, it turns out that the contribution to~\eqref{eq:condh1h2} from~\eqref{eq:sumB1B2} is also at most the quantity in~\eqref{eq:contribution_B1CB2C}. 
Since the analysis of~\eqref{eq:sumB1B2} is analogous to that of~\eqref{eq:sumB1CB2C} done above, we omit it for brevity. 

Putting everything together, in particular~\eqref{eq:condh1h2}, the cases~\eqref{eq:sumB1B2}---\eqref{eq:sumB1CB2C}, and the corresponding contributions~\eqref{eq:contributionB1B2C} and~\eqref{eq:contribution_B1CB2C}, 
and recalling that both of these contributions should be counted twice, 
we obtain that 
\begin{align}
\p \left( \wh{v}_{\ML} = v^{*} \, \middle| \, A_{12} \right) 
&\leq \frac{2}{t_{1} + 1} + \frac{2}{t_{2} + 1} + \frac{2}{\min \left\{ t_{1} , t_{2} \right\} + 1} + \frac{2/3}{\max \left\{ t_{1}, t_{2} \right\} + 1} - \frac{8}{\left( t_{1} + 1 \right) \left( t_{2} + 1 \right)} \notag \\
&\leq \frac{6 + 2/3}{\min \left\{ t_{1}, t_{2} \right\} + 1} - \frac{8}{\left( t_{1} + 1 \right) \left( t_{2} + 1 \right)}. \label{eq:probMLEcorrectA12_simple}
\end{align}

Now finally putting together~\eqref{eq:MLE_total_prob_odd_odd},~\eqref{eq:probMLEcorrectA12C_simple}, and~\eqref{eq:probMLEcorrectA12_simple}, 
we obtain that 
\begin{align*}
\p \left( \wh{v}_{\ML} = v^{*} \right) 
&\leq \frac{d-1}{d} \left( \frac{6 + 2/3}{\min \left\{ t_{1}, t_{2} \right\} + 1} - \frac{8}{\left( t_{1} + 1 \right) \left( t_{2} + 1 \right)} \right) 
+ \frac{1}{d} \cdot \frac{11}{\left( t_{1} + 1 \right) \left( t_{2} + 1 \right)} \\
&= \frac{d-1}{d} \cdot \frac{6 + 2/3}{\min \left\{ t_{1}, t_{2} \right\} + 1} 
+ \frac{11 - 8 \left( d - 1 \right)}{d\left( t_{1} + 1 \right) \left( t_{2} + 1 \right)}.
\end{align*}
Since $d \geq 3$, we have that $11 - 8 (d - 1) < 0$, so the second term above is negative. 
Therefore
\[
\p \left( \wh{v}_{\ML} = v^{*} \right) 
\leq 
\frac{d-1}{d} \cdot \frac{6 + 2/3}{\min \left\{ t_{1}, t_{2} \right\} + 1},
\]
which concludes the proof. 
\end{proof}

\end{document}